\documentclass[a4paper,UKenglish]{lipics}
\usepackage[T1]{fontenc}

\usepackage{microtype}%if unwanted, comment out or use option "draft"
\usepackage{booktabs}
\usepackage{graphicx}
\usepackage{cleveref}
\usepackage{wrapfig}
\usepackage{xcolor}
\usepackage{multicol}
\newcommand{\two}{${}^2$}
\newcommand{\twoBD}{${}^2_\mathrm{BD}$}
\newcommand{\const}{\mathrm{const}}
\newcommand{\Uni}{\mathrm{Uni}}
\newcommand{\voc}{\mathrm{voc}}
\newcommand{\universal}{\mathrm{\forall\forall}}
\newcommand{\existential}{\mathrm{\forall\exists}}
\newcommand{\As}{\mathrm{S}}
\newcommand{\att}{\textbf{att}{}}
\newcommand{\domunb}{\textbf{dom}${}^{\mathrm{U}}${}}
\newcommand{\dom}{\textbf{dom}}
\newcommand{\dombind}[1]{\textbf{dom}_{#1}^{\mathrm{B}}{}}
\newcommand{\domunbMM}{\textbf{dom}^{\mathrm{U}}{}}
\newcommand{\Bl}{\textbf{bool}}

\newcommand{\Dom}{\textbf{dom}}
\newcommand{\codes}{\textbf{codes}}
\newcommand{\Blshort}{\textbf{b}}
\newcommand{\Domshort}{\textbf{d}}
\newcommand{\codesshort}{\textbf{c}}
\newcommand{\relnames}{\textbf{relnames}{}}
\newcommand{\sqlvars}{\textbf{SQLvars}{}}
\newcommand{\connames}{\textbf{connames}{}}
\newcommand{\srt}{\textbf{sort}}
\newcommand{\e}[1]{\textbf{\textit{#1}}}
\newcommand{\cI}{\mathcal{I}}
\newcommand{\bfR}{\mathbf{R}}
\newcommand{\ar}{\mathrm{ar}}
\newcommand{\mcA}{\mathcal{A}}

\newcommand{\nil}{\mathtt{nil}}
\newcommand{\pre}{\mathrm{pre}}
\newcommand{\post}{\mathrm{post}}
\newcommand{\bnd}{\mathit{bnd}}

\definecolor{dkgreen}{rgb}{0,.6,0}
\definecolor{dkblue}{rgb}{0,0,.6}
\definecolor{dkyellow}{cmyk}{0.3,0.3,1,.5}
\definecolor{dark-gray}{gray}{0.5}

\newtheorem{proposition}{Proposition}
\title{On the automated verification of web applications with embedded SQL}
\author[1]{Shachar Itzhaky}
\author[2]{Tomer Kotek}
\author[3]{Noam Rinetzky}
\author[3]{Mooly Sagiv}
\author[3]{Orr Tamir}
\author[2]{Helmut Veith}
\author[2]{Florian Zuleger}
\affil[1]{Massachusetts Institute of Technology
}
\affil[2]{Vienna University of Technology
 }
\affil[3]{Tel Aviv University
 }
\authorrunning{S. Itzhaky, T. Kotek, N. Rinetzky, M. Sagiv,  O. Tamir, H. Veith, and F. Zuleger} %mandatory. First: Use abbreviated first/middle names. Second (only in severe cases): Use first author plus 'et. al.'

\Copyright{Shachar Itzhaky, Tomer Kotek, Noam Rinetzky, Mooly Sagiv,  Orr Tamir, Helmut Veith, and Florian Zuleger}%mandatory, please use full first names. LIPIcs license is "CC-BY";  http://creativecommons.org/licenses/by/3.0/

% mandatory: Please choose ACM 1998 classifications from http://www.acm.org/about/class/ccs98-html . E.g., cite as "F.1.1 Models of Computation". 
% mandatory: Please provide 1-5 keywords
\serieslogo{}%please provide filename (without suffix)
\volumeinfo%(easychair interface)
  {Billy Editor and Bill Editors}% editors
  {2}% number of editors: 1, 2, ....
  {Conference title on which this volume is based on}% event
  {1}% volume
  {1}% issue
  {1}% starting page number
\EventShortName{}
\DOI{10.4230/LIPIcs.xxx.yyy.p}% to be completed by the volume editor
\begin{document}

\maketitle

\begin{abstract}
A large number of web applications is based on a relational database
together with a program, typically a script, that enables the user to
interact with the database through embedded SQL queries and commands. In this
paper, we introduce a method for formal automated verification of such
systems which connects database theory to mainstream program
analysis. We identify a fragment of SQL which captures the behavior
of the queries in our case studies, is algorithmically decidable, and
facilitates the construction of weakest preconditions. Thus, we can
integrate the analysis of SQL queries into a program analysis tool
chain. To this end, we implement a new decision procedure for the SQL
fragment that we introduce. We demonstrate practical applicability of
our results with three case studies, a web administrator, a
simple firewall, and a conference management system. 
 \end{abstract}

\newcommand{\sqlFragment}{{SmpSQL}}
\newcommand{\progLang}{{SmpSL}}
\newcommand{\FOTwo}{{FO\two}}

%running example
\newcommand{\nwlEmail}{{NS}}
\newcommand{\nwlAdmin}{\texttt{T}}
\newcommand{\nwlEmailGhost}{\ensuremath{\nwlEmail_\mathit{gh}}}
\newcommand{\nwlAdminGhost}{\ensuremath{\nwlAdmin_\mathit{gh}}}
\newcommand{\nwl}{nwl}
\newcommand{\mail}{user}
\newcommand{\admin}{admin}
\newcommand{\subscribed}{subscribed}
\newcommand{\confirmcode}{code}

\newcommand{\activeUser}{\ensuremath{\mathit{active}}}
\newcommand{\newsl}{n}
\newcommand{\user}{u}
\newcommand{\code}{c}
\newcommand{\cd}{cd}
\newcommand{\newcode}{\mathtt{new}\texttt{\_}\mathtt{code}}

\newcommand{\subscribeghost}{sub_{gh}}

\newcommand{\unsubscribe}{\texttt{unsubscribe}}
\newcommand{\subscribe}{\texttt{subscribe}}
\newcommand{\confirm}{\texttt{confirm}}
\newcommand{\condition}{\ensuremath{\varphi}}
\newcommand{\precondition}{\condition_{pre}}
\newcommand{\postcondition}{\condition_{post}}

\newcommand\SELECT{\texttt{SELECT}~}
\newcommand\INSERT{\texttt{INSERT}~}
\newcommand\DELETE{\texttt{DELETE}}
\newcommand\UPDATE{\texttt{UPDATE}~}
\newcommand\FROM{~\texttt{FROM}~}
\newcommand\INTO{~\texttt{INTO}~}
\newcommand\SET{~\texttt{SET}~}
\newcommand\WHEREnospace{\texttt{WHERE}~}
\newcommand\WHERE{~\texttt{WHERE}~}
\newcommand\IN{~\texttt{IN}~}
\newcommand\AND{~\texttt{AND}~}
\newcommand\OR{~\texttt{OR}~}
\newcommand\NOT{\texttt{NOT}~}
\newcommand\CHOOSE{\texttt{CHOOSE}~}

\newcommand\cond{\mathit{cond}}
\newcommand\vempty{\mathtt{empty}}

\newcommand\semp[1]{[\![{#1}]\!]}
\newcommand\wlp[1]{\mathrm{wp}\semp{#1}}
\newcommand\eqdef{\;\widehat{=}\;}

\newcommand{\mypara}[1]{\subparagraph*{#1.}}

%\renewcommand\strut[1]{\rule{0pt}{#1}}

%====================================================

\section{Introduction}\label{se:intro}

Web applications are often written in a scripting language such as PHP and store their data in a relational database
which they access using SQL queries and data-manipulating commands~\cite{williams2004web}. 
This combination facilitates fast development of web applications, which exploit the reliability and efficiency of the underlying database engine and use the flexibility of the script language to interact with the user.
While the database engine is typically a mature software product with few if any severe errors, 
the script with the embedded SQL statements does not meet the same standards of quality.
%Time constraints, non-expert programmers, the complexity of the resulting system, etc. often lead to errors in the resulting script.

With a few exceptions (such as~\cite{DBLP:journals/sigmod/DeutschHV14,DBLP:journals/vldb/FernandezFLS00}) the systematic analysis of programs with embedded-SQL statements has been a blind spot in both the database and the computer-aided verification community.
The verification community has mostly studied the analysis of programs which fall into two classes: programs with (i) numeric variables and complex control structure, (ii) complex pointer structures and objects; however, the modeling of data and their relationships has not received the same attention.
Research in the database community on the other hand has traditionally focused on correct design of databases rather than correct use of databases.

This paper lays the ground for an interdisciplinary methodology which extends the realm of program analysis to programs with embedded SQL.
Since the seminal papers of Hoare, the first step for developing program analysis techniques is a precise mathematical framework for defining programming semantics and correctness.
In this paper we develop a Hoare logic for a practically useful but simple fragment of SQL, called \sqlFragment, and a simple scripting language, called \progLang, which has access to \sqlFragment\ statements.
Specifically, we describe a decidable logic for formulating specifications and develop a weakest precondition calculus for \progLang\ programs; thus our Hoare logic allows to automatically discharge verification conditions.
When analyzing \progLang\ programs, we treat SQL as a black box library whose semantics is given by database theory.
Thus we achieve verification results relative to the correctness of the underlying database engine.

We recall from Codd's theorem~\cite{codd1972relational} that the core of SQL is equivalent in expressive power to first-order logic FO.
Thus, it follows from Trakhtenbrot's theorem~\cite{ar:Trakhtenbrot} that it is undecidable whether an SQL query guarantees a given post condition.
We have therefore chosen our SQL fragment \sqlFragment\ such that it captures an interesting class of SQL commands, 
but corresponds to a decidable fragment of first-order logic, namely FO\twoBD, the restriction of first-order logic in which all variables aside from two range over 
fixed finite domains called \emph{bounded domains}. 
The decidability of the finite satisfiability problem of FO\twoBD{} follows from that of FO\two, the fragment of first-order logic which uses only two variables. 
Although the decidability of FO\two\ was shown by Mortimer~\cite{ar:Mortimer} and a complexity-wise tight decision procedure was later described by Gr\"{a}del, Kolaitis and Vardi~\cite{GKV}, 
we provide the first efficient implementation of finite satisfiability of FO\two.

We illustrate our methodology on the example of a simple web administration tool based on~\cite{panda}. 
The PANDA web administrator is a simple public domain web administration tool written in PHP.
We describe in Section~\ref{se:running} how the core mailing-list administration functionality falls into the scope of \progLang.
%In particular, we present hand-translated  functions for (un-)subscribing users to mailing-lists in \progLang.
We formulate a specification consisting of a database invariant and pre- and postconditions.
% Applying our methodology
% we found an error in the source code of the PANDA webserver. We also formulate
% a corrected implementation of the PANDA webserver.
Our framework allows us to automatically check the correctness of such specifications using our own FO\twoBD\ reasoning tool. 

%%%%%%%%%% BUG
% We illustrate our methodology on the example of the PANDA web administrator~\cite{panda}. 
% The PANDA web administrator is a simple public domain webserver written in PHP.
% We describe in Section \ref{se:running} how the core mailing-list administration functionality falls into our programming language \progLang.
% We formulate a specification consisting of a database invariant and pre- and postconditions.
% Applying our methodology
% we found an error in the source code of the PANDA webserver. We also formulate
% a corrected implementation of the PANDA webserver.
% Our framework allows us to automatically check whether the specification holds using our own FO\twoBD\ reasoning tool. 

\paragraph*{Main contributions}

\begin{enumerate}

\item We define \sqlFragment, an SQL fragment which is contained in FO\twoBD.

\item We define a a simple imperative script language \progLang\ with embedded \sqlFragment\ statements. 

\item We give a construction for weakest preconditions in FO\twoBD\ for \progLang.

\item We implemented the weakest precondition computation for \progLang.

\item We implemented a decision procedure for FO\twoBD. 
The procedure is based on the decidability and NEXPTIME completeness result 
for FO\two\ by~\cite{GKV}, but we use a more involved algorithm which reduces the problem to a SAT solver
and is optimized for performance. 
\end{enumerate}

%%%%%%%%%% BUG
We evaluate our methodology on three applications: a web administrator, a simple firewall, 
and a conference management system. 
%Our tool can prove or disprove properties such as correct subscription to mailing lists or correct communication between devices with respect to firewall restrictions. 
We compared our tool with Z3~\cite{de2008z3}, currently the most advanced general-purpose SMT solver with (limited) support for quantifiers. 
    In general, our tool performs better than Z3  in several examples for checking the validity of verification conditions of \progLang\ programs. 
    However, our tool and Z3 have complementary advantages: 
    Z3 does well for unsatisfiable instances while our tool performs better on satisfiable instances.
We performed large experiments with custom-made blown up versions of the web administrator and the firewall examples, which suggest that our tool scales well.
Moreover, we tested the scalability of our approach by comparing of our underlying FO\two\ solver with three solvers on a set of benchmarks 
we assembled inspired by combinatorial problems. The solvers we tested against are
Z3, the SMT solver CVC4~\cite{CVC4}, and the model checker Nitpick~\cite{blanchette2010nitpick}. Our solver outperformed each of these solvers on some of the benchmarks. 

%----------------------------------------------------
\section{Running Example}
\label{se:running}

We introduce our approach on the example of a simple web service. The example is a translation from PHP with embedded SQL commands into \progLang\
of code excerpts 
from the Panda web-administrator.
The web service provides several services implemented in dedicated functions for subscribing a user to a newsletter, 
deleting a newsletter, making a user an admin of a newsletter, sending emails to all subscribed users of a newsletter, etc.
We illustrate our verification methodology by exposing an error in the Panda web-administrator. The verification methodology we envision in this paper consists of 
(1) maintaining database invariants and (2) verifying a contract specification for each function of the web service. 

The database contains several tables including $\mathit{\nwlEmail}=\mathit{NewsletterSubscription}$
with attributes $\mathit{\nwl}$, $\mathit{\mail}$, $\mathit{\subscribed}$ and $\mathit{\confirmcode}$. 
The database is a structure whose universe is partitioned into three sets: $\Dom^{\mathrm{U}}$, $\Bl^{\mathrm{B}}$, and $\codes^{\mathrm{B}}$. 
The attributes $\mathit{\nwl}$ and $\mathit{\mail}$ range over the finite set $\Dom^\mathrm{U}$, the attribute $\mathit{\subscribed}$
ranges over $\Bl^{\mathrm{B}}=\{\mathit{true},\mathit{false}\}$, and the attribute $\mathit{\confirmcode}$ ranges over the fixed finite set $\codes^{\mathrm{B}}$. 
The superscripts in $\Dom^{\mathrm{U}}$, $\Bl^{\mathrm{B}}$, and $\codes^{\mathrm{B}}$ serve to indicate that the domain $\Dom^{\mathrm{U}}$ is unbounded, while the Boolean domain 
and the domain of codes are bounded (i.e.\ of fixed finite size). 
When $s=\mathit{true}$, $(n,u,s,c) \in \mathit{\nwlEmail}$ signifies that the \emph{user} $u$ is \emph{subscribed} to the \emph{newsletter} $n$. 
The process of being (un)subscribed from/to a newsletter
requires an intermediary confirmation step in which the confirm code $c$ plays a role.

Figure~\ref{fig:code} provides the functions $\subscribe$, $\unsubscribe$, and $\confirm$ translated manually into \progLang. \footnote{The reader may
wish to compare the \progLang\ implementation of $\confirm$ to the PHP implementation in PANDA, provided in Appendix~\ref{app:panda:confirm}. 
}
The comments in quotations {\tt \color{dark-gray} // ``$\ldots$''} originate from the PHP source code. 
The intended use of these functions is as follows: 
In order to subscribe a user $u$ to a newsletter $n$, the function $\subscribe$  
is called with inputs $n$ and $u$ (for example by a web interface operated by the newsletter admin or by the user). 
$\subscribe$ stores the tuple $(n,u,\mathit{false},new\_code)$ in $\mathit{\nwlEmail}$, where $new\_code$ is a confirmation code which does not 
occur in the database, 
and an email containing a confirmation URL is sent to the user $u$. 
Visiting the URL triggers a call to $\confirm$ with input $new\_code$, which subscribes $u$ to $n$ by replacing 
the tuple $(n,u,\mathit{false},new\_code)$ of $\mathit{\nwlEmail}$ to with $(n,u,\mathit{true},\nil)$. 
For $\unsubscribe$ the process is similar, and crucially, $\unsubscribe$ uses the same $\confirm$ function.
$\confirm$ decides whether to subscribe or unsubscribe according to whether $n$ is currently subscribed to $u$.
The \CHOOSE\ command selects one row non-deterministically.

The database preserves the invariant 
\begin{equation}\label{eq:Inv}
\begin{array}{lll}
\mathit{Inv}
%&=& \forall_{\Domshort} x,y.\, \exists_{\Blshort}^{\leq 1} s.\, \exists_{\codesshort}^{\leq 1} c.\,\mathit{\nwlEmail}(x,y,s,c)
&=& \forall_{\Domshort} x,y.\, \forall_{\Blshort} s_1,s_2.\, \forall_{\codesshort} c_1,c_2.\,\left(
 (s_1=s_2 \land c_1=c_2)\lor\bigvee_{i=1,2}\neg\mathit{\nwlEmail}(x,y,s_i,c_i)
 \right)
\end{array}
\end{equation}
%$\mathit{Inv}_1$ 
$\mathit{Inv}$
says that the pair $(n,u)$ of newsletter and user is a key of the relation $\mathit{\nwlEmail}$. 
The subscripts of the quantifiers denote the domains over which the quantified variables range. 
In our verification methodology we add invariants as additional conjuncts to the pre- and post-conditions of every function.
In this way invariants strengthen the pre-conditions and can be used to prove the post-conditions of the functions.
On the other hand, the post-conditions require to re-establish the validity of the invariants.

\begin{figure}\tt
\footnotesize
\noindent
%\{\,$\forall (x,y). \nwlAdmin(x,y) \leftrightarrow \nwlAdminGhost(x,y) \wedge$\\
%\hspace*{0.1cm} $\forall (x,y). \nwlEmail(x,y) \leftrightarrow \mathit{\nwlEmailGhost}(x,y)$\,\}\\
\subscribe(\newsl,\user): \\
\hspace*{0.5cm} A = \SELECT *\FROM \nwlEmail \WHERE \mail\ = \user\ AND \nwl\ = \newsl;\\
\hspace*{0.5cm} if (A != empty)  exit; {\color{dark-gray}// "This address is already registered to this newsletter."} \\
\hspace*{0.5cm} \INSERT (\newsl,\user,false,$\newcode$)\INTO \nwlEmail;\\
\hspace*{0.5cm} {\color{dark-gray}// Send confirmation email to \user}\\
%\{$\,\forall (x,y). \nwlAdmin(x,y) \leftrightarrow \nwlAdminGhost(x,y) \wedge$\\
%\hspace*{0.1cm} $\forall (x,y). \nwlEmail(x,y) \leftrightarrow (\mathit{\nwlEmailGhost}(x,y) \wedge$\\
%\hspace*{2.8cm}$(\nwlAdmin(n,active) \rightarrow x \neq n \vee y \neq u))$\,\}
\ \\
\unsubscribe(\newsl,\user):  \\
\hspace*{0.5cm} A = \SELECT *\FROM \nwlEmail \WHERE \mail\ = \user\ AND \nwl\ = \newsl;\\
\hspace*{0.5cm} if (A = empty) exit; {\color{dark-gray}// "This address is not registered to this newsletter."} \\
\hspace*{0.5cm} \UPDATE \nwlEmail \SET \confirmcode\ = $\newcode$\WHERE \mail\ = \user\ AND \nwl\ = \newsl\\
\hspace*{0.5cm} {\color{dark-gray}// Send confirmation email to \user}\\
\ \\
\confirm(\cd): \\
\hspace*{0.5cm} A = \SELECT \subscribe \FROM \nwlEmail \WHERE \confirmcode\ = \cd;\\
\hspace*{0.5cm} if (A = empty) exit;  {\color{dark-gray}//"No such code"} \\
\hspace*{0.5cm} s1\ = \CHOOSE A;\\
\hspace*{0.5cm} if (s1 = false) 
        \UPDATE \nwlEmail \SET \subscribed\ = true, \confirmcode\ = $\nil$\WHERE \confirmcode\ = \cd\\
\hspace*{0.5cm}  else\ \DELETE \FROM \nwlEmail \WHERE \confirmcode\ = \cd;\\
\caption{\label{fig:code} Running Example: \progLang\ code. }
\end{figure}
Figure~\ref{fig:conditions} provides pre- and post-conditions $\pre_{\mathtt{f}}$ and $\post_{\mathtt{f}}$ for each of the three functions $\mathtt{f}$.
The relation names $\Domshort$, $\Blshort$, and $\codesshort$ are interpreted as the sets  $\Dom^{\mathrm{U}}$, $\Bl^{\mathrm{B}}$, and $\codes^{\mathrm{B}}$, respectively. 
Proving correctness amounts to proving the correctness of each of the \emph{Hoare triples} 
$\{\pre_{\mathtt{f}}\land \mathit{Inv}\}\ \mathtt{f}\ \{\post_{\mathtt{f}} \land \mathit{Inv}\}$. Each Hoare triple specifies a contract:
after every execution of $\mathtt{f}$, the condition $\post_{\mathtt{f}} \land \mathit{Inv}$ should be satisfied if
$\pre_{\mathtt{f}} \land \mathit{Inv}$ was satisfied before executing $\mathtt{f}$.
$\pre_\subscribe$ and $\pre_\unsubscribe$ express that $\mathit{new\_code}$ is an unused non-nil code and that $\mathit{\nwlEmailGhost}$
is equal to $\mathit{\nwlEmail}$. $\mathit{\nwlEmailGhost}$ is a \emph{ghost table}, used in the post-conditions to relate the state before the execution of the function to the state
after the execution. $\mathit{\nwlEmailGhost}$ does not occur in the functions and is not modified. 
$\post_\subscribe$ and $\post_\unsubscribe$ express that $\mathit{\nwlEmail}$ is obtained from $\mathit{\nwlEmailGhost}$ by inserting or updating
a row satisfying $\tt \mail\ = \user\ \mathtt{AND}\ \nwl\ = \newsl$
whenever the \texttt{exit} command is not executed. 
The intended behavior of $\confirm$ depends on which function created  $\cd$.
$\pre_\confirm$ introduces a Boolean ghost variable $\subscribeghost$
whose value is true (respectively false) if $\cd$ was generated as a new code in \subscribe\ (respectively \unsubscribe). 
$\subscribeghost$
does not occur in $\confirm$. $\post_\confirm$ express that, when $\subscribeghost$ is true, $\mathit{\nwlEmail}$ is obtained from $\mathit{\nwlEmail}$ by 
toggling the value of the column $\it \subscribed$ from false to true in the $\mathit{\nwlEmailGhost}$ row whose confirm code is $\cd$; 
when $\subscribeghost$ is false, $\mathit{\nwlEmail}$ is obtained from $\mathit{\nwlEmailGhost}$ by deleting the row with confirm code $\cd$. 
% Observe that specification belongs to FO\twoBD, the fragment of first-order logic in which the only variables ranging over $\Domshort$ are $x$ and $y$,
% see Section~\ref{se:ver-fotwo}

Let us now describe the error which prevents $\confirm$ from satisfying its specification. 
Consider the following scenario. First, $\subscribe$ is called and then $\unsubscribe$, both with the same input $\newsl$ and $\user$.
Two confirm codes are created: $c_{\mathtt{s}}$ by $\subscribe$ and $c_{\mathtt{u}}$ by $\unsubscribe$. 
At this point, $\mathit{\nwlEmail}$ contains a single row for the newsletter $n$ and user $u$ namely $(\newsl,\user,\mathit{false},c_\mathtt{u})$. 
The user receives two confirmation emails containing the codes $c_\mathtt{s}$ and $c_\mathtt{u}$. 
Clicking on the confirmation URL for $c_\mathtt{s}$ (i.e.\ running $\confirm(c_\mathtt{s})$) 
has no effect since $c_\mathtt{s}$ does not occur in the database. 
However, clicking on the confirmation URL for $c_\mathtt{s}$ results in subscribing $\user$ to $\newsl$. This is an error, since 
confirming a code created in $\unsubscribe$ should not lead to a subscription. 

Our tool automatically checks whether the program satisfies its specification. 
If not, the programmer or verification engineer may try to
refine the specification to adhere more closely to the intended behavior (e.g.\ by adding an invariant). 
In this case, the program is in fact incorrect, so no meaningful correct specification can be written for it.

In Section~\ref{se:verification} we describe a \emph{weakest-precondition calculus} $\wlp{\cdot}$ 
which allows us to automatically derive the weakest precondition for a post-condition with regard to a \progLang\ program.
For our example functions $\mathtt{f}$, $\wlp{\cdot}$ allows us to automatically derive
$\wlp{\mathtt{f}}\post_\mathtt{f}$.
The basic property of the weakest precondition is that $\post_\mathtt{f}$ holds after $\mathtt{f}$ has executed iff
$\wlp{\mathtt{f}}\post_\mathtt{f}$ held immediately at the start of the execution. 
It then remains to show that the pre-condition $\pre_\mathtt{f}$ implies  $\wlp{\mathtt{f}}\post_\mathtt{f}$.
This amounts to checking the validity of the verification conditions $\mathit{VC}_\mathtt{f}=\pre_\mathtt{f}\rightarrow \wlp{\mathtt{f}}\post_\mathtt{f}$.

Our reasoner for FO\two\ sentences is the back-end for our verification tool. 
The specification in this example is all in FO\twoBD.
%with one caveat.%namely $\mathit{Inv}_2$. However, 
%$\mathit{Inv}_2$ can be rewritten in FO\twoBD\ using nested quantifications on the same variable $z\in\{x,y\}$, see Appendix~\ref{ap:invariants}. 
The weakest precondition of a \progLang\ program applied to a FO\twoBD\ sentence gives again a FO\twoBD sentence. Hence $\mathit{VC}_\mathtt{f}$
are all in FO\twoBD. 
Automatically deciding the validity of FO\twoBD\ sentences using our FO\two\ decision procedure is described in Section~\ref{se:FO2}.
Recall that $\codes^{\mathrm{B}}$ is of fixed finite size. Here
$|\codes^{\mathrm{B}}| = 3$ is sufficient to detect the error. Observe that the same confirm code may be reused
once it is replaced with $\nil$ in $\confirm$, so the size of the database is unbounded. 
The size of $\codes^{\mathrm{B}}$ must be chosen manually when applying our automatic tool.

A simple way to correct the error in \confirm\ is by adding $\subscribeghost$ as a second argument of
\confirm\ and replacing $\tt if\ (s1 = false)\ \cdots$ with 
$\tt if\ (\subscribeghost = false)\ \cdots$. Since $s_1$ is no longer used, the CHOOSE command can be deleted. 
The value of $\subscribeghost$ received by \confirm\ is set correctly by \subscribe\ and \unsubscribe. 
With these changes, the error is fixed and \confirm\ satisfies its specification. 
In the scenario from above, the call to \confirm\ with $c_\mathtt{s}$ and $\subscribeghost = \mathit{true}$ 
leaves the database unchanged, while the call to  
\confirm\ with $c_\mathtt{u}$ and $\subscribeghost = \mathit{false}$ deletes the row $(n,u,\mathit{false},c_\mathtt{u})$. 

%%%%%%%%%% BUG The results of verification of \subscribe, \unsubscribe, and \confirm, both before and after correcting the error, are described in Section~\ref{se:experiments}. 

\begin{figure}
\small
\[
 \begin{array}{lll}
\pre_\mathtt{g} &=& \mathit{\nwlEmail} = \mathit{\nwlEmailGhost} \land \mathrm{good}\mbox{-}\mathrm{code}(\mathit{new\_code})
\\
\mathrm{good}\mbox{-}\mathrm{code}(c') &=& \codesshort(c')\land (c'\not=\mathit{nil}) \land\forall_{\Domshort} x,y.\,\forall_{\Blshort} s.\, \neg \mathit{\nwlEmail}(x,y,s,c')
\\
\post_\mathtt{g} &=& \forall_{\Domshort} x,y.\, \forall_{\Blshort} s.\, \forall_{\codesshort} c.\, 
\mathit{\nwlEmail}(x,y,s,c) \leftrightarrow (\varphi_{\mathtt{g},1}\lor \varphi_{\mathtt{g},2})
\\
\varphi_\mathtt{\subscribe,1} &=& \mathit{\nwlEmailGhost}(x,y,s,c)
\\
\varphi_\mathtt{\subscribe,2} &=& (n = x) \land (u = y)\land (s = \mathit{false}) \land 
(c=\mathit{new\_code})
\\
&&
\land 
\neg \exists_{\Blshort} s'.\, \exists_{\codesshort} c'.\, \mathit{\nwlEmailGhost}(n,u,s',c')
\\
\varphi_{\mathtt{\unsubscribe},1} &=& (n \not= x) \land (u \not= y)\land \mathit{\nwlEmailGhost}(x,y,s,c)
\\
\varphi_{\mathtt{\unsubscribe},2} &=& (n = x) \land (u = y)\land 
(c=\mathit{new\_code})
\land 
 \exists_{\codesshort} c'.\, \mathit{\nwlEmailGhost}(n,u,s,c')
\\
\\
\pre_\confirm &=& \mathit{\nwlEmail} = \mathit{\nwlEmailGhost} \land \Blshort(\subscribeghost)
\\
\post_\confirm &=& \bigwedge_{\mathit{tt}\in \Blshort}
\subscribeghost = \mathit{tt} \to \left(\forall_\Domshort x,y.\,\forall_\Blshort s.\,\forall_\codesshort c.\,
\mathit{\nwlEmail}(x,y,s,c) \leftrightarrow \psi_{\mathit{tt}}\right)
\\
\psi_{\mathit{false}}& = & \cd \not= c \land \mathit{\nwlEmailGhost}(x,y,s,c)
\\
\psi_{\mathit{true}} &=& \cd \not= c \land \mathit{\nwlEmailGhost}(x,y,s,c) 
\lor ( c=\mathit{nil} \land s = \mathit{true} \land \mathit{\nwlEmailGhost}(x,y,\mathit{false},\cd))
\end{array}
\]

 \caption{\label{fig:conditions} Running Example: Pre- and post-conditions. $\mathtt{g}$ is either $\subscribe$ or $\unsubscribe$. }
\end{figure}

%----------------------------------------------------
\section{Verification of \progLang\ Programs}\label{se:languages}

Here we introduce our programming language and our verification methodology. 
We introduce the SQL fragment  \sqlFragment\ in Section~\ref{se:sqlFragment}
and the scripting language \progLang\ in Section~\ref{se:progLang}. 
In Section~\ref{se:verification} we explain the weakest precondition transformer of \progLang, and 
we show how discharging verification conditions of FO\twoBD\ specification reduces to reasoning in FO\two. 

\subsection{The SQL fragment \sqlFragment}\label{se:sqlFragment}

\subsubsection{Data model of \sqlFragment}\label{se:datamodel-sql}

The data model of \sqlFragment\ is based on the presentation of the relational model in Chapter 3.1  of~\cite{abiteboul1995foundations}. 
%Chapter 7.1 on SQL. 
We assume finite sets of $\dombind{1},\ldots,\dombind{s}$ called the 
\emph{bounded domains} and an infinite set \domunb{} called the \emph{unbounded domain}. 
The domains are disjoint. 
We assume three disjoint countably infinite sets: the set of attributes \att, the set of relation names \relnames,
and the set of variables \sqlvars. We assume a function $\srt:\att\to \{\domunbMM,\dombind{1},\ldots,\dombind{s}\}$. 
A \e{table} or a \e{relation schema} is a relation name and a finite sequence of attributes.
The attributes are the names of the columns of the table. The \e{arity} $\ar(R)$ of a relation schema $R$ is the number of its attributes. 
A \e{database schema} is a non-empty finite set of tables.
%Every database schema contains a special constant $\null$. 

A \e{database instance} $\cI$
of a database schema $\bfR$ is a many-sorted structure with finite domains $\dom_0 \subseteq \domunbMM$
and $\dom_j = \dombind{j}$ for $1\leq j\leq s$. We denote by $\srt_\cI$ the function obtained from $\srt$
by setting $\srt_\cI(\mathit{att}) = \dom_0$ whenever $\srt(\mathit{att}) = \domunbMM$.  
The relation schema $R =(\mathit{relname},\mathit{att}_1,\ldots,\mathit{att}_e)$ is interpreted in $\cI$ as a relation 
$R^\cI \subseteq
\srt_\cI(\mathit{att}_1)\times \cdots\times \srt_\cI(\mathit{att}_e)$.
A \e{row} is a tuple in a relation $R^\cI$.

A database schema $\bfR$ is \textbf{valid} for \sqlFragment\ if for all relation schemas $R$ with attributes $\mathit{att}_1,\ldots,\mathit{att}_e$ 
in $\bfR$, there are at most two attributes $\mathit{att}_j$ for which $\srt(\mathit{att}_j) = \domunbMM$. 
In the sequel we assume that all database schemas are valid.  
The \sqlFragment\ commands will be allowed to use variables from \sqlvars. 
%The set \sqlvars{} is the disjoint union of $\sqlvars_0,\ldots,\sqlvars_s$,
%which range over $\dom_0,\ldots,\dom_s$, respectively. 
We denote members of $\sqlvars$ by $p$, $p_1$, etc. 

\subsubsection{Queries in \sqlFragment}

Given a relation schema $R$ and attributes $att_{1},\ldots,att_{n}$ of $R$,
the syntax of SELECT is:
\[\tt
\begin{array}{lll}
\tt \left\langle Select\right\rangle  &::= &  \tt \SELECT  att_{a_1},\ldots,att_{a_i}  \FROM  R\ \WHEREnospace \left\langle Condition\right\rangle\\
\tt \left\langle Condition\right\rangle &::= & 
\tt 
 att_{b_1},\ldots,att_{b_j}  \IN \left\langle Select\right\rangle | 
 \left\langle Condition\right\rangle  \AND \left\langle Condition\right\rangle | \\
{} & {}  & 
\tt\left\langle Condition\right\rangle  \OR \left\langle Condition\right\rangle |\, 
 \NOT \left\langle Condition\right\rangle |\, att_m=p
 \end{array}
\]
where $p$ is a variable and $1\leq m, a_1,\ldots,a_i,b_1,\ldots,b_j \leq n$.
The semantics of $\left\langle Select\right\rangle$
is the set of tuples from the projection of $R$ on $att_{a_1},\ldots,att_{a_i}$
which satisfy $\left\langle Condition\right\rangle$. The condition $\tt att_m=p$ indicates that
the set of rows of $R$ in which the attribute $\mathit{att}_m$
has value $p$ is selected. The condition $\tt {att}_{b_1},\ldots,att_{b_i}  \IN \left\langle Select\right\rangle$ selects the set of rows of $R$ in which 
$\mathit{att}_{b_1},\ldots,att_{b_i}$ are mapped
to one of the tuples  queried in the nested query $\left\langle Select\right\rangle$.

\subsubsection{Data-manipulating commands in \sqlFragment}

\sqlFragment\ supports the three primitive commands INSERT, UPDATE,
and DELETE. 

Let $R$ be a relation schema with attributes $att_{1},\ldots,att_{n}$.
Let $p,p_{1},\ldots,p_n$ be variables from \sqlvars. The syntax
of the primitive commands is:
\begin{eqnarray*}
\tt \left\langle Insert\right\rangle & ::= & \tt \INSERT (p_{1},\ldots,p_n)  \INTO  R\\
\tt \left\langle Update\right\rangle & ::= & \tt \UPDATE  R  \SET  att_m=p\
\WHEREnospace \left\langle Condition\right\rangle\\
\tt \left\langle Delete\right\rangle & ::= & \tt \DELETE  \FROM  R\
  \WHEREnospace \left\langle Condition\right\rangle
\end{eqnarray*}
The semantics of INSERT, UPDATE and DELETE is given in the natural
way. We allow update commands which set several attributes simultaneously. 
We assume that the data manipulating commands are used in a domain-correctness fashion,
i.e.\ \INSERT\ and \UPDATE\ may only assign values from $\srt(\mathit{att}_k)$ to 
any attribute $\mathit{att}_k$.

\subsection{The script language \progLang}\label{se:progLang}
\subsubsection{Data model of \progLang}\label{se:datamodel-proglang}
The data model of \progLang\ extends that of \sqlFragment\ with constant names and additional relation schemas. 
We assume a countably infinite set of constant names \connames, which is disjoint from $\att, \domunbMM,\dombind{1},\ldots,\dombind{s},\relnames$ but contains
\sqlvars.

A \e{state schema} is a database schema $\bfR$ expanded with a tuple of constant names 
$\overline{const}$. 
A \e{state} interprets a state schema. It consists of a database instance $\cI$ expanded with a tuple
of universe elements
$\overline{const}^\cI$ interpreting 
$\overline{const}$. 
In programs, the constant names play the role of local variables, domain constants (e.g.\ $\mathit{true}$ and $\mathit{true}$)
and of inputs to the program\footnote{We deviate from~\cite{abiteboul1995foundations} in the treatment of constants in that we do not assume that constant names are always interpreted as 
\emph{distinct}  members of \domunb. This is so since several program variables or inputs can have the same value.}.

\subsubsection{\progLang\ programs}

The syntax of \progLang\ is given by
\[
 \begin{array}{lll}
  \tt \left\langle Program \right\rangle &::= &\tt \left\langle Command \right\rangle  |  \left\langle Program \right\rangle; \left\langle Command \right\rangle \\
  \tt \left\langle Command \right\rangle &::= & \tt  \left\langle  Insert \right\rangle | \left\langle  Update \right\rangle |  \left\langle Delete \right\rangle |\, R = \left\langle  Select \right\rangle | 
   \, \bar{d} = \CHOOSE\ R\ |\\
  &&\tt if\ (cond)\,\left\langle Program \right\rangle\, else\, \left\langle Program \right\rangle  |\ if\ (cond)\,  exit
 \end{array}
\]

Every data-manipulating command $C$ of \sqlFragment\ is a \progLang\ command. 
The semantics of $C$ in \progLang\ is the same as in \sqlFragment, with the caveat that 
the variables receive their values from their interpretations (as constant names) in the state, and 
$C$ is only legal if all the variables of $C$ indeed appear in the state schema as constant names. 

The command $\tt R =  \left\langle Select \right\rangle$ 
assigns  the result of a \sqlFragment\ query to a relation schema $R \in \bfR$ whose arity and attribute sorts match the select query. 
Executing the command in a state $(\cI,\overline{const}^\cI)$ %such that the database schema of $\cI$ is $\bfR$
sets $R^\cI$ to the relation selected by $S$, leaving the interpretation of all other names unchanged. 
The variables in the query receive their values from their interpretations in the state, and 
for the command to be legal, all variables in the query must appear in the state schema as constant names. 

Given a relation schema $R\in \bfR$ with attributes $att_1,\ldots,att_n$ and a tuple $\bar{d} = (d_1,\ldots,d_n)$ of constant names from $\overline{const}$, 
 $\tt \bar{d} = \CHOOSE R$ is a \progLang\ command. If $R^\cI$ is empty,
 the command has no effect. If $R^\cI$ is not empty, executing this command
 sets $(d_1^\cI,\ldots,d_n^\cI)$
to the value of a 
non-deterministically selected row from $R^\cI$.

The branching commands have the natural semantics. 
Two types of branching conditions $\tt cond$ are allowed:
$(R = \vempty)$ and $(R ~!{=}~ \vempty)$, which check whether $R^\cI$ is the empty set, and 
 $(c_1 = c_2)$ and $(c_1 != c_2)$, which check whether $c_1^\cI = c_2^\cI$.

 See Fig.~\ref{fig:code} for examples of \progLang\ programs.

\subsection{Verification of \progLang\ programs}\label{se:verification}

\subsubsection{SQL and FO}
It is well-established that a core part of SQL is captured by FO by Codd's classical theorem relating the expressive power of relational algebra to relational calculus. 
While  SQL goes beyond  FO in several aspects, such as aggregation, grouping, and arithmetic operations (see~\cite{libkin2003expressive}), these aspects are not allowed in \sqlFragment.
Hence, FO is especially suited for reasoning about \sqlFragment\ and \progLang. 

The notions of state schema and state fit naturally in the syntax and semantics of FO. 
In the sequel, a \e{vocabulary} is a tuple of relation names and constant names. 
For a FO-formula $\psi$, we write $\voc(\psi)$ for the vocabulary consisting of the relation names and constant names in $\psi$. 
Every state schema $\bfR$ is a vocabulary. 
A state $(\cI,\overline{const}^\cI)$ interpreting a state schema $\bfR$ and a tuple of constant names 
$\overline{const}$
is an $\left\langle \bfR, \overline{const} \right\rangle$-structure.

\subsubsection{Hoare verification of \progLang\ programs and weakest precondition}\label{se:hoare}\label{se:vc}
\begin{figure}
\footnotesize
\[
\begin{array}{@{}l@{}}
\begin{array}{lcl}
  \semp{\tt att_i=c}^R & \eqdef & v_i=c \\
  \semp{\tt att_{b_1},\ldots,att_{b_j}\ \mathtt{IN\ }S_1}^R & \eqdef & \semp{S_1}[v_{b_k}/v_k: 1\leq k\leq j] \\
  \semp{\tt cond_1\ \mathtt{AND\ }cond_2}^R & \eqdef & \semp{\cond_1}^R \land \semp{\cond_2}^R \\
  \semp{\tt cond_1\ \mathtt{OR\ }cond_2}^R & \eqdef & \semp{\cond_1}^R \lor \semp{\cond_2}^R \\
  \semp{\tt \mathtt{NOT\ }cond_1}^R & \eqdef & \lnot\semp{\cond_1}^R
\end{array} 
\\
\\
\begin{array}{l}
  \semp{\tt \SELECT att_{a_1},\ldots,att_{a_i} \FROM R \WHERE cond} \eqdef 
  (\exists v_{a_{i+1}},\ldots,v_{a_{n}}R(\bar{v}) \land \semp{\cond}^R)[v_\ell/v_{a_\ell}: 1\leq \ell\leq i]\\
  \mbox{where }\{a_1,\ldots,a_n\} = \{1,\ldots,n\}
  \\
  \\
\end{array}\\
\begin{array}{lll}
  \wlp{\tt \INSERT (c_1,\ldots,c_n) \INTO R}Q &\eqdef &
    Q\big[R(\bar{\alpha})\lor \bigwedge_{i=1}^{n} \alpha_i=c_i \ \big/\  R(\bar{\alpha})\big] \\
    \wlp{\tt \DELETE\FROM R \WHERE cond}Q &\eqdef &
    Q\big[R(\bar{\alpha})\land \lnot\semp{cond}^R[\alpha_i/v_i: 1\leq i\leq n] \ \big/\  R(\bar{\alpha})\big] \\
\wlp{\tt \UPDATE R \SET att_j=c \WHERE cond}Q &\eqdef& 
    Q\big[R(\bar{\alpha})\land \lnot\semp{cond}^R[\alpha_i/v_i: 1\leq i\leq n] \lor{} \\
    &&~~~~\exists v_j R(\overline{\alpha^j})\land \semp{cond}^R[\alpha^j_i/v_i: 1\leq i\leq n] \land \alpha_j=c \\
  &&\multicolumn{1}{r}{
     \qquad\qquad\big/\  R(\bar{\alpha})\big]} \\
\end{array}
\end{array}
\]
\caption{\label{weakestpre-rules-sisql}	Rules for weakest precondition for {\sqlFragment} basic commands.
We denote by $R$ a relation schema with attributes $\langle att_1,\ldots,att_n\rangle$.  
We write $\alpha^j_i $ for $\alpha_i$ if $i \neq j$, and for $v_i$ if $i = j$. We denote $\bar{v}=(v_1,\ldots,v_n)$,
$\bar{\alpha}=(\alpha_1,\ldots,\alpha_n)$,
and $\overline{\alpha^j} = (\alpha_1^j,\ldots,\alpha_n^j)$.
Note that each of the last three rows $ Q\big[\mathit{expr}(\alpha) \big/ R(\bar{\alpha})\big] $
substitutes every occurrence of $R$ with an updated expression $\mathit{expr}$.
}
\end{figure}
Hoare logic is a standard program verification methodology~\cite{hoare1969axiomatic}. 
Let $P$ be a \progLang\ program and let $\varphi_{\mathit{pre}}$ and $\varphi_{\mathit{post}}$ be FO-sentences. 
A \e{Hoare triple} is of the form
$\{\varphi_{\mathit{pre}}\} P \{\varphi_{\mathit{post}}\}$.
A Hoare triple is a \emph{contract} relating the state before the program is run with the state afterward.
The goal of the verification process is to prove that the contract is correct. 

Our method of proving that a Hoare triple is valid reduces the problem to that of finite satisfiability of a FO-sentence.
We compute the \e{weakest precondition} $\wlp{P}\varphi_{\mathit{post}}$ of $\varphi_{\mathit{post}}$ with respect to the program $P$. 
The weakest precondition transformer was introduced in Dijkstra's classic paper~\cite{dijkstra1975guarded}, c.f.~\cite{jhala2009software}. 
Let $\mcA_P$ denote the state after executing $P$ on the initial state $\mcA$.
The main property of the weakest precondition is:
 $
  \mcA_P\models \varphi_{\mathit{post}}\mbox{ iff } \mcA \models \wlp{P}\varphi_{\mathit{post}}
 $.
Using $\wlp{\cdot}$ we can rephrase the problem of whether the Hoare triple $\{\varphi_{\mathit{pre}}\} P \{\varphi_{\mathit{post}}\}$ is valid
in terms of FO reasoning on finite structures:
\emph{Is the FO-sentence 
$ \varphi_{\mathit{pre}} \to \wlp{P}\varphi_{\mathit{post}}$
a tautology?}
Equivalently, 
\emph{is the FO-sentence
$
 \varphi_{\mathit{pre}} \land \neg \wlp{P}\varphi_{\mathit{post}}
$
unsatisfiable? }
Section~\ref{se:ver-fotwo}
discusses the resulting FO reasoning task. 

We describe the computation of the weakest precondition inductively
for \sqlFragment\ and \progLang.  The weakest precondition for {\sqlFragment} is given in Fig.~\ref{weakestpre-rules-sisql},
 and for  {\progLang} in Fig.~\ref{weakestpre-rules-sisl}.
For \sqlFragment\ conditions, $\semp{\cdot}^R$ is a formula with $n$ free first-order variables $v_{1},\ldots,v_n$ for a conditional
expression in the context of relation schema $R$ of arity $n$.
$\semp{\tt \SELECT\cdots\, \FROM\ R\, \cdots\,}$ is also a formula with free variables $v_{1},\ldots, v_n$ describing
the rows selected by the SELECT query.
The rules $\wlp{s}Q$ transform a (closed) formula $Q$, which is a postcondition of the command
$s$, into a (closed) formula expressing the weakest precondition.
The notation $\psi[t/v]$ indicates substitution of all free occurrences of the variable
$v$ in $\psi$ by the term $t$. 

The notation
$\psi\big[\theta(\alpha_1,\ldots,\alpha_n)/R(\alpha_1,\ldots,\alpha_n)]$ indicates 
that
any atomic sub-formula of $\psi$ of the form $R(\alpha_1,\ldots,\alpha_n)$ (for any $\alpha_1,\ldots,\alpha_n$) 
is replaced by $\theta(\alpha_1,\ldots,\alpha_n)$ (with the same $\alpha_1,\ldots,\alpha_n$). The formula
$\theta(v_1,\ldots,v_n)$ has $n$ free variables, and $\theta(\alpha_1,\ldots,\alpha_n)$
is obtained by substituting each $v_i$ into $\alpha_i$. The $\alpha_i$ may be variables or constant names. 

The weakest precondition of a \progLang\ program is obtained by applying the weakest precondition of its commands. 

\begin{figure}
\footnotesize
\[
\begin{array}{lcl}
  \semp{\tt c_1=c_2} & \eqdef & c_1=c_2 \\
  \semp{\tt c_1~!{=}~c_2} & \eqdef & c_1 \neq c_2 \\
  \semp{\tt R~!=\vempty} & \eqdef & \exists v_1,\ldots,v_nR(v_1,\ldots,v_n) \\
  \semp{\tt R{=}~\vempty} & \eqdef & \lnot\exists v_1,\ldots,v_nR(v_1,\ldots,v_n)
\end{array} 
\]
\[
\renewcommand\arraystretch{1.2}
\begin{array}{lll}
  \wlp{\tt R = \SELECT\ \cdots\ }Q &\eqdef& Q\big[\semp{\SELECT\ \cdots\ }(\alpha_1,\ldots,\alpha_n)\big/R(\alpha_1,\ldots,\alpha_n)\big] \\
  \wlp{\tt (d_1,\ldots,d_n) = \CHOOSE R}Q &\eqdef& \forall u_1,\ldots,u_n \big(R(u_1,\ldots,u_n) \to Q[u_i/d_i: 1\leq i\leq n]\big) \\
  \wlp{\tt if\ cond~ s_1~ else\ s_2}Q &\eqdef& (\lnot\semp{\cond} \land \wlp{s_2}Q) \lor (\semp{\cond} \land \wlp{s_1}Q)
\end{array}
\]
\caption{\label{weakestpre-rules-sisl}
	Rules for weakest precondition construction for {\progLang} basic commands. 
        The weakest precondition of 
	$\tt if\ cond~ exit; s_2$
	is the same as that of $\tt if\ !cond~s_2$. }
\end{figure}

\subsubsection{The specification logic \texorpdfstring{FO\twoBD}{FO2BD}\ and decidability of verification}\label{se:ver-fotwo}

As discussed in Section~\ref{se:hoare}, using the weakest precondition, the problem of verifying 
Hoare triples can be reduced to the problem of checking satisfiability
of a FO-sentence by a finite structure. While this problem
is not decidable in general by Trakhtenbrot's theorem, it is decidable for a fragment of FO we denote FO\twoBD,
 which extends the classical {\em two-variable fragment} FO\two.
The logic FO\two\ is the set of all FO formulas
which use only variables the variables $x$ and $y$. 
The vocabularies of FO\two-sentences are not allowed function names, only relation and constanst names. 
Note FO\two\ cannot express that a relation name is interpreted as a function.  
FO\two{} contains the equality symbol $=$. 
FO\twoBD\ extends FO\two\ by allowing quantification on an unbounded number of variables, 
under the restriction that all variables besides from $x$ and $y$ range over the bounded domains only. 
%We denote the bounded domain variables of FO\twoBD\ by $x^j,z^j,x_1^j$, etc. where $j$ is the index of the bounded domain $\dombind{j}$ on which
%these variables range. 

FO\twoBD\ is the language of our invariants and pre- and postconditions, see Eq.~(\ref{eq:Inv}) and Fig.~\ref{fig:conditions} in Section~\ref{se:running}. 
An important property of FO\twoBD\ is that it is essentially closed under taking weakest precondition 
according to Figs.~\ref{weakestpre-rules-sisql} and~\ref{weakestpre-rules-sisl}
since all relation schemas in a (valid) database schema have at most $2$ attributes whose sort is \domunb. 
We reduce the task of reasoning over FO\twoBD\ to reasoning over FO\two.

\begin{theorem}
Let $\{\varphi_{\mathit{pre}}\} P \{\varphi_{\mathit{post}}\}$ be a Hoare triple such that both $\varphi_{\mathit{pre}}$ and $\varphi_{\mathit{post}}$ belong to FO\twoBD. 
The problem of deciding whether $\{\varphi_{\mathit{pre}}\} P \{\varphi_{\mathit{post}}\}$  is valid is decidable. 
\end{theorem}
\begin{proof}[(sketch)]
By Section~\ref{se:hoare}, $\{\varphi_{\mathit{pre}}\} P \{\varphi_{\mathit{post}}\}$  is valid iff 
$
 \theta = \neg(\varphi_{\mathit{pre}} \land \neg \wlp{P}\varphi_{\mathit{post}})
$
is satisfiable by a finite structure. 
We take the simplifying assumption that in all the tables, the sort of the first and second attributes $\mathit{att}_1$ and $\mathit{att}_2$
is $\dom^{\mathrm{U}}$.
This assumption does not effect the expressive power of \progLang.
Examination of the weakest precondition rules in 
Figs.~\ref{weakestpre-rules-sisql} and~\ref{weakestpre-rules-sisl}
reveals that
the only variables ranging over the unbounded domain are $v_1$ and $v_2$. 
Let $\theta'$ be the FO\twoBD\ sentence obtained from $\theta$ by substituting $v_1$ and $v_2$ with $x$ and $y$ respectively, and restricting the range of the quantifiers 
appropriately:
for a command manipulating or querying a table $R$ with attributes $\mathit{att}_1,\ldots,\mathit{att}_n$ in Figs.~\ref{weakestpre-rules-sisql} and~\ref{weakestpre-rules-sisl},  
each quantifier $\forall v_k$ or $\exists v_k$ 
is replaced with $\forall_{\srt(\mathit{att}_k)} v_k$ or $\exists_{\srt(\mathit{att}_k)} v_k$.
We compute an FO\two\ sentence $\theta''$ which is equivalent to $\theta'$ by hard-coding the bounded domains.
Every table $T$ which contains an attribute $\mathit{att}$ with $\srt(\mathit{att}) = \dombind{j}$ of size $d$ is replaced with $d$ tables $T_1,\ldots,T_d$
which do not have the attribute $\mathit{att}$. This change is reflected in $\theta''$, e.g.\ existential quantification is replaced with disjunction. 
By the decidability of finite satisfiability of FO\two-sentences, we get that the problem of deciding whether 
$\{\varphi_{\mathit{pre}}\} P \{\varphi_{\mathit{post}}\}$  is valid is decidable. 
\end{proof}

%----------------------------------------------------
\section{\texorpdfstring{FO\two{}}{FO2} Reasoning}\label{se:FO2}\label{se:foreasonig}

\subsection{The bounded model property of \texorpdfstring{FO\two}{FO2}}\label{se:bounded}
Section~\ref{se:FO2} is devoted to our algorithm for FO\two\ finite satisfiability.
The main ingredient for this algorithm is the \emph{bounded model property}, which guarantees that if an FO\two$(\tau)$ sentence $\phi$
over vocabulary $\tau$
is satisfiable by any $\tau$-structure -- finite or infinite -- it is satisfiable by a finite $\tau$-structure whose cardinality is bounded by a computable function 
of $\phi$. 
The bound guaranteed in the first decidability proof of the finite satisfiability problem by Mortimer \cite{ar:Mortimer} was doubly exponential in the size of the formula.
Later, Gr\"{a}del, Kolaitis and Vardi \cite{GKV} proved the \emph{exponential model property}, from which we get that the problem is NEXPTIME-complete. 
The naive NEXPTIME algorithm arising from the exponential model property amounts to computing the exponential bound $\bnd(\phi)$
from \cite{GKV}, 
non-deterministically guessing $t\leq \bnd(\phi)$ and a $\tau$-structure $\mathcal{A}$
with universe $\{1,\ldots,t\}$, checking whether $\mathcal{A}$ satisfies $\phi$, and answering accordingly. 
Since the truth-value of FO-sentences is invariant to $\tau$-isomorphisms, a $\tau$-structure of cardinality at most $\bnd(\phi)$ satisfies $\phi$
iff such a structure with universe $\{1,\ldots,t\}$, $t\leq \bnd(\phi)$, satisfies $\phi$. 
Appendix~\ref{app:se:efficient} discusses a more refined version of the bound from \cite{GKV}. 

\subsection{Finite satisfiability using a SAT solver} \label{se:sat}

Our algorithm for FO\two{} finite satisfiability reduces the problem of finding a satisfying model of cardinality bounded by $\bnd$
to the satisfiability of a propositional Boolean formula in Conjunctive Normal Form CNF, which is then solved using a SAT solver.
The bound in \cite{GKV} is  given for formulas in Scott Normal Form (SNF) only. 
We use a refinement of SNF we call \emph{Skolemized Scott Normal Form (SSNF)}. 
The CNF formula we generate encodes the semantics of the sentence $\psi$ on a structure
whose universe cardinality is bounded by $\bnd$. An early precursor for the use of a SAT solver for finite satsifiability is~\cite{Mccune94adavis-putnam}. 

\subsubsection{Skolemized Scott Normal Form} 
An FO\two-sentence is in \emph{Skolemized Scott Normal Form} if it is of the form
\begin{gather}
\label{eq:SSNF}
 \forall x \forall y\, \left(\alpha(x,y) \land \bigwedge_{i=1}^m F_i(x,y)\to\beta_i(x,y)\right) \land \bigwedge_{i=1}^m \forall x \exists y\, F_i(x,y)
\end{gather}
where $\alpha$ and $\beta_i$, $i=1,\ldots,m$, are quantifier-free formulas which do not contain any $F_j$, $j=1,\ldots,m$. Note that $F_i$ are relation names. 
\begin{proposition}\label{prop:skolemized-scott}
Let $\tau$ be a vocabulary and $\phi$ be a FO\two$(\tau)$-sentence. There are polynomial-time computable vocabulary $\sigma\supseteq \tau$ and FO\two$(\sigma)$-sentence
$\psi$ such that 
 (a) $\psi$ is in SSNF;
 (b) The set of cardinalities of the models of $\phi$ is equal to the corresponding set for $\psi$; and 
 (c) The size of $\psi$ is linear in the size of $\phi$.
\end{proposition}

Proposition~\ref{prop:skolemized-scott} follows from the discussion before Proposition 3.1 in~\cite{GKV}, by applying
an additional normalization step converting SNF sentences to SSNF sentences.\footnote{The word Skolemized is used in reference to the standard \emph{Skolemization} process of eliminating existential quantifiers by introducing fresh function names called 
\emph{Skolem functions}. In our case,
since function names are not allowed in our fragment, we introduce the relation names $F_i$, to which we refer as \emph{Skolem relations}. 
Moreover, we cannot eliminate the existential quantifiers entirely, but only simplify the formulas in their scope to the atoms $F_i(x,y)$. 
}\footnote{The linear size of $\psi$
uses our relation symbols have arity at most $2$ to get rid of a $\log$ factor in~\cite{GKV}. }

\subsubsection{The CNF formula}\label{se:FO2:CNF}
Given the sentence $\psi$ in SSNF 
from Eq.~(\ref{eq:SSNF})
and a bound $\bnd(\psi)$, we build a CNF propositional Boolean formula $C_\psi$ 
which is satisfiable iff $\psi$ is satisfiable. The formula $C_\psi$ will serve as the input to the SAT solver. 
%The construction of $C_\psi$ is described in this subsection. 
First we construct a related CNF formula $B_\psi$.
The crucial property of $B_\psi$ is that it is satisfiable
iff $\psi$ is satisfiable by a model \emph{of cardinality exactly $\bnd(\psi)$}. 

It is convenient to assume $\psi$ does not contain constants.
If $\psi$ did contain constants $c$, they could be replaced by unary relations $U_c$ of size $1$.
Being an unary relation of size $1$ is definable in FO\two. 
Any atom containing $c$ cannot use both $x$ and $y$, and hence the universe member interpreting $c$ can be quantified: 
e.g.\ $R(x,c)$ is replaced with $\exists y\,U_c(y)\land R(x,y)$. 
Let $\const(\psi)$ be the set of unary relations $U_c$ corresponding to constants. 

We start by introducing the variables and clauses which guarantee that $B_\psi$ encodes a structure with the universe $\{1,\ldots,\bnd(\psi)\}$. 
Later, we will add clauses to guarantee that this structure satisfies $\psi$. 
For every unary relation name $U$ in $\psi$ 
and $\ell_1\in \{1,\ldots,\bnd(\psi)\} $, let $v_{U,\ell_1}$ be a propositional variable. 
For every binary relation name $R$ in $\psi$
and $\ell_1,\ell_2 \in \{1,\ldots,\bnd(\psi)\}$, let $v_{R,\ell_1,\ell_2}$ be a propositional variable. 
The variables $v_{U,\ell_1}$ and $v_{R,\ell_1,\ell_2}$ encode the interpretations of the unary and binary relation names $U$ and $R$ in the straight-forward way
(defined precisely below). 
Let $V_\psi$ be the set of all variables $v_{U,\ell_1}$ and $v_{R,\ell_1,\ell_2}$.  

Given an assignment ${\As}$ to the variables of $V_\psi$ we define 
the unique structure $\mathcal{A}_{\As}$ as follows:
\begin{enumerate}
 \item The universe $A_{\As}$ of $\mathcal{A}_{\As}$ is $\{1,\ldots,\bnd(\psi)\}$;
 \item An unary relation name $U$ is interpreted as the set $\{\ell_1 \in A_{\As} \mid {\As}(v_{U,\ell_1})=True\}$;
 \item A binary relation name $R$ is interpreted as the set $\{(\ell_1,\ell_2) \in A_{\As}^2 \mid {\As}(v_{R,\ell_1,\ell_2})=True\}$;
\end{enumerate}
For every structure $\mathcal{A}$ with universe $\{1,\ldots,\bnd(\psi)\}$, 
 there is ${\As}$ such that $\mathcal{A} = \mathcal{A}_{\As}$. 

 Before defining $B_\psi$ precisely we can already state the crucial property of $B_\psi$:

\begin{proposition}\label{prop:assignment}~
 %\item For every ${\As}$, ${\As}$ satisfies $B_\psi$ iff $\mathcal{A}_{\As}$ satisfies~$\psi$. 
 $\psi$ is satisfiable by a structure with universe $\{1,\ldots,$ $\bnd(\psi)\}$ iff $B_\psi$ is satisfiable. 
\end{proposition}

The formula $B_\psi$
is the conjunction of $B^{\mathrm{eq}}$, $B^{\forall\exists}$, and $B^{\forall\forall}$,
described in the following.

\vspace{5pt}
\noindent
\textbf{The equality symbol.}
The equality symbol requires special attention. 
% For every $1\leq \ell_1 \not=\ell_2\leq \bnd(\psi)$,
% let $B^{\mathrm{eq}}_{\ell_1,\ell_2} = \neg v_{=,\ell_1,\ell_2}$
% and for every $1\leq \ell_1\leq \bnd(\psi)$, let 
% $B^{\mathrm{eq}}_{\ell_1}=v_{=,\ell_1,\ell_1}$. 
Let 
$$B^{\mathrm{eq}} = \bigwedge_{1\leq \ell_1\not=\ell_2\leq m} (\neg v_{=,\ell_1,\ell_2}) \land \bigwedge_{1\leq \ell\leq m} v_{=,\ell,\ell}$$
$B^{\mathrm{eq}}$ enforces that the equality symbol is interpreted correctly as the equality relation on universe elements.

\vspace{5pt}
\noindent
\textbf{The $\forall \exists$-conjuncts.}
For every conjunct $\forall x \exists y\, F_i(x,y)$ and $1 \leq \ell_1 \leq \bnd(\psi)$, 
let $B^{\forall\exists}_{i,\ell_1}$ be the clause
$\bigvee_{\ell_2=1}^{\bnd(\psi)} v_{F_i,\ell_1,\ell_2}$. This clause
says that
there is at least one universe element $\ell_2$ such that $\mathcal{A}_{\As}\models F(\ell_1,\ell_2)$. 
Let 
$$B^{\forall\exists} = \bigwedge_{1\leq i\leq m}\bigwedge_{1\leq \ell_1\leq \bnd(\psi)} B^{\forall\exists}_{i,\ell_1}$$
For every truth-value assignment ${\As}$ to $V_{\psi}$, $\mathcal{A}_{\As}$ satisfies~$\bigwedge_{i=1}^m \forall x \exists y\, F_i(x,y)$
iff ${\As}$ satisfies $B^{\forall\exists}$.

\vspace{5pt}
\noindent
\textbf{The $\forall\forall$-conjunct.}
Let $\forall x \forall y\, \alpha'$ be the unique $\forall\forall$-conjunct of $\psi$. 
For every $1\leq \ell_1,\ell_2\leq  \bnd(\psi)$, let $\alpha''_{\ell_1,\ell_2}$
denote the propositional formula obtained from the quantifier-free FO\two\ formula $\alpha'$ 
by substituting every atom $a$ with the corresponding propositional variable for $\ell_1$ and $\ell_2$
as follows:\\ 
$
\begin{array}{lllllllll}
U(x)&\mapsto& v_{U,\ell_1},& \ \ \ \  \ \ 
R(y,y)&\mapsto& v_{R,\ell_2,\ell_2},& \ \ \ \ \ \ 
R(x,x)&\mapsto& v_{R,\ell_1,\ell_1}\\ 
U(y)&\mapsto& v_{U,\ell_2},& \ \ \ \  \ \ 
R(x,y)&\mapsto& v_{R,\ell_1,\ell_2},& \ \ \ \  \ \ 
R(y,x)&\mapsto& v_{R,\ell_2,\ell_1} \ \  \ \  \ \
\end{array}
$\\
Let $B^{\forall\forall}_{\ell_1,\ell_2}$ be the Tseitin transformation of $\alpha''_{\ell_1,\ell_2}$
to CNF \cite{tseitin1983complexity}, see also \cite[Chapter 2]{biere2009handbook}. The Tseitin transformation
introduces a linear number of new variables of the form $u^{\gamma}_{\ell_1,\ell_2}$, one for each
sub-formula $\gamma$ of $\alpha''_{\ell_1,\ell_2}$. The transformation
guarantees that, 
for every assignment $\As$ of $V_\psi$, $\As$ satisfies $\alpha''_{\ell_1,\ell_2}$ iff $\As$ can be expanded to satisfy 
$B^{\forall\forall}_{\ell_1,\ell_2}$. 
Let $$B^{\forall\forall} = \bigwedge_{1\leq \ell_1,\ell_2\leq  \bnd(\psi)} B^{\forall\forall}_{\ell_1,\ell_2}(\ell_1,\ell_2)$$
Appendix~\ref{app:se:forallforall} gives the construction of the CNF formula $B^{\forall\forall}$ according to the Tseitin transformation explicitly.

The construction of $B_\psi$ is finished and Proposition~\ref{prop:assignment} holds. 
Note that \cite{GKV} guarantees only that  $\bnd(\psi)$ is an \emph{upper bound} on the cardinality of a satisfying model. 
Therefore, we build a formula $C_\psi$ based on $B_\psi$ such that  
$C_\psi$ is satisfiable iff $\psi$ is satisfiable by a structure of cardinality \emph{at most} $\bnd(\psi)$. 
We leave the technical details of the construction of $C_{\psi}$ to the appendix. 
The algorithm for finite satisfiability of a FO\two-sentence $\phi$ consists of computing the SSNF $\psi$ of $\phi$ and returning the result of 
a satisfiability check using a SAT solver on $C_{\psi}$. 
Both the number of variables and the number of clauses in $C_{\Uni(\psi)}$ are quadratic in $\bnd(\psi)$.

%----------------------------------------------------
\section{Experimental Results}\label{se:experiments}
\subsection{Details of our tools}
The verification condition generator described in Section~\ref{se:vc} is implemented in Java, JFlex and CUP. It is employed to parse
the schema, precondition and postcondition and the \progLang\ programs. 
The tool checks that the pre and post conditions are specified
in FO\two{} and that the scheme is well defined. 
The SMT-LIB v2~\cite{smt2} standard
language is used as the output format of the verification condition generator.
We compare the behavior of our FO\two-solver with Z3 on the verification condition generator output. The validity
of the verification condition can be checked by providing its negation to the SAT solver. 
If the SAT solver exhibits a satisfying assignment then that serves as counterexample for the correctness of the program. 
If no satisfying assignment exists, then the generated verification condition is valid, and therefore the program satisfies the assertions.
The FO\two-solver described in Section~\ref{se:FO2} is implemented in python and uses pyparsing to parse the SMT-LIB v2~\cite{smt2} file.
The FO\two-solver assumes a FO\two-sentence as input and uses \emph{Lingeling}~\cite{Biere10lingeling} SAT solver as a base Solver.

%Python, SMT-LIB v2, ...
%pyparsing for SMT-LIB v2
% EPR, PLY

\subsection{Example applications}
%\begin{wrapfigure}{r}{7cm}
\begin{table}
  \centering
    \begin{tabular}[t]{lrr}
    \toprule
                 & FO\two-solver & Z3 \\
    \midrule
                web-subscribe      & 2.62s   & 0.08s \\
               web-unsubscribe    & 0.779s  & OM \\
               firewall           & 0.876s  & OM \\
               conf-bid           & 0.451s  & 0.015s \\
               conf-assign        & 0.369s  & 0.013s \\
               conf-display       & 0.992s  & 0.016s \\
    \bottomrule
    \multicolumn{3}{c}{incorrect}
    \end{tabular}%
    \quad
    \begin{tabular}[t]{lrr}
    \toprule
                 & FO\two-solver & Z3 \\
    \midrule
       web-subscribe     & 1.07s   & 0.1s \\
       web-unsubscribe   & 8.209s  & 0.1s \\
       firewall          & 2.82s   & 0.103s \\
       conf-bid          & TO      & 0.22s \\
       conf-assign       & 1.196s  & 0.2s  \\
       conf-display      & TO      & 0.16s \\
    \bottomrule
    \multicolumn{3}{c}{correct}
    \end{tabular}%
    \caption{Running time comparison for example benchmarks}
  \label{tab:smallExamples}%
%\end{wrapfigure}%
\end{table}
%%%%%%%%%%%%%%% BUG

We tried our approach with a few programs inspired by real-life applications.
The first case study is a simplified version of the newsletter functionality
included in the PANDA web administrator, that was already discussed and is
shown in Fig.~\ref{fig:code}.\footnote{
We omit the confirmation step due to a missing feature in the implementation of the weakest precondition, 
however the final version of the tool will support the code from Table~\ref{fig:code}.}
The second is an excerpt from a firewall that updates a table of which device
is allowed to send packets to which other device.
See Appendix~\ref{app:network} for the code and specifications of the firewall. 
The third is a conference management system with a database of papers, and
transactions to manage the review process: reviewers first \emph{bid} on
papers from the pool of submissions, with a policy that a users cannot bid
for papers with which they are conflicted. The chair then \emph{assigns} reviewers
to papers by selecting a subset of the bids. At any time, users can ask to
\emph{display} the list of papers, with some details, but the system may hide
some confidential information, in particular, users should not be able to see
the status of papers before the program is made public. 
We show how our system detects an information flow bug in which  
the user might learn that some papers were accepted prematurely by examining the session assignments. 
This bug is based on a bug we observed in a real system. 
See Appendix~\ref{app:conference} for the code and specifications of the conference management system. 
%Our two example programs is simplified versions of the \emph{subscribe} and \emph{unsubscribe}
%functions of the PANDA web administrator, omitting the confirmation step. 
%The third example program is the removal of a network device from the database of a \emph{firewall}.
Each example comes with two specifications, one correct and the other incorrect.

The running time in seconds for all of our examples is reported in Table~\ref{tab:smallExamples}. Timeout is set to 60 minutes and denoted as TO. If the solver reaches \emph{out of memory} we mark it as OM. 
On the set of correct examples, both solvers answer within a few seconds, Z3 terminates within milliseconds,
while FO\two-solver takes a few seconds and times out on some of them.
On the set of incorrect examples, Z3 fails to answer while our solver performs well.
Note that correct examples correspond to unsatisfiable FO\two-sentences, while 
incorrect examples correspond to satisfiable FO\two-sentences.

\subsection{Examining scalability}\label{se:inflated}
\mypara{Inflated examples}
In order to evaluate scalability to large examples we inflated our base examples. 
For instance, while the \emph{subscribe} example from Table~\ref{tab:smallExamples} consisted of the subscription of one new email to a mailing-list, 
Table~\ref{tab:inflatedExamples} presents analogous examples in which multiple emails are subscribed to multiple mailing-lists.
The column \emph{multiplier} details the number of individual subscriptions in each example program.
The \emph{unsubscribe} and \emph{firewall} example programs are inflated similarly (see Appendix~\ref{app:network-inflated}).  

We have tested both our FO\two-solver and Z3 on large examples and the results reported in Table~\ref{tab:inflatedExamples}. 
The high-level of the results is similar to the case of the small examples. 
On the incorrect examples set Z3 continues to fail mostly due to running out 
of memory, though it succeeds on the subscribe example. 
On the correct examples set Z3 continues to outperform the FO\two-solver. 

\mypara{Artificial examples}
In addition, we constructed a set of artificial benchmarks comprising of several
families of FO\two-sentences. Each family is parameterized by a number that controls
the size of the sentences (roughly corresponding to the number of quantifiers in the sentence).
These problems are inspired by combinatorial problems such as graph coloring and paths.
We ran experiments using the FO\two-solver and three publicly available solvers:
Z3, CVC4 (which are SMT solvers), and Nitpick (a model checker).
The results are collected in Table~\ref{tab:artificial}.
The artificial benchmarks are available at~\url{http://forsyte.at/wp-content/uploads/artificial-smt2.tar.gz}. 

\mypara{Scalability of FO\two-solver}
We shall conclude that the FO\two-solver, despite being a proof of concept in python with minimal optimizations, handles well incorrect specifications 
(satisfiable sentences)
and also scales well on them. However it struggles on the correct specifications and does not scale well. This suggests that in future work we may choose to 
run both our solver and Z3 in parallel and answer according the first answer obtained. We also intend to explore how to improve the performance 
of our solver in the case of incorrect examples. 
By construction, whenever FO\two-solver finds a satisfying model, its size is at most $4$ times that of the minimal model. 
(The constant $4$ can be decreased or increased. ) 

\begin{table}\centering
\footnotesize
    \begin{tabular}{llrrrrrr}
    \toprule
          &            & \multicolumn{3}{c}{FO\two-solver} & \multicolumn{3}{c}{Z3} \\
          & multiplier & 1 & 10 & 100 & 1 & 10 & 100 \\
    \midrule
    incorrect
    & subscribe   & 2.62s & 0.973s & 4.04s        &   0.08s & 0.126s & 0.203s    \\
    & unsubscribe & 0.779s & 0.529s & 1.27s       &   OM & OM & OM   \\
    & firewall    & 0.876s & 0.723s & 2.251s      &   OM & OM & OM \\
    correct
    & subscribe   & 1.07s & 98.249s & TO          &   0.1s & 0.11s & 0.116s \\
    & unsubscribe & 8.209s & 456.308s & TO        &   0.1s & 0.157s & 0.201s \\
    & firewall    & 2.82s & 50.142s & 2951.882s   &   0.103s & 0.121s & 0.143s \\
%     incorrect & subscribe & 1     & 2.62s  & TO \\
%           &               & 10    & 0.973s & 0.126s \\
%           &               & 100   & 4.04s  & 0.203s \\
%           & unsubscribe & 1     & 0.779s & OM \\
%           &             & 10    & 0.529s & OM \\
%           &             & 100   & 1.27s  & OM \\
%           & firewall & 1     & 0.876s & OM \\
%           &          & 10    & 0.723s & OM \\
%           &          & 100   & 2.251s & OM \\
%     correct & subscribe & 1     & 1.07s & 0.1s \\
%           &             & 10    & 98.249s & 0.11s \\
%           &             & 100   & TO & 0.116s \\
%           & unsubscribe & 1     & 8.209s & 0.1s \\
%           &             & 10    & 456.308s & 0.157s \\
%           &             & 100   & TO & 0.201s \\
%           & firewall & 1     & 2.82s  & 0.103s \\
%           &          & 10    & 50.142s & 0.121s \\
%           &          & 100   & 2951.882s & 0.143s \\
    \bottomrule
    \end{tabular}%
    \caption{Running time comparison on inflated examples}
  \label{tab:inflatedExamples}%
\end{table}

\begin{table}\centering
\footnotesize
\newcommand\TO{\multicolumn{1}{c}{TO}}
\newcommand\hsep{\hline\rule{0pt}{1em}}
    \begin{tabular}{lc|l|rrrr}
    \toprule
          &   size    & status & \multicolumn{1}{c}{Z3} & \multicolumn{1}{c}{CVC4} & \multicolumn{1}{c}{Nitpick} & \multicolumn{1}{c}{FO\two-solver} \\
    \midrule
2col	&	3	&	unsat	&	0m0.037s	&	0m0.076s	&	\TO	&	\TO	\\	
	&	4	&	sat	&	\TO	&	\TO	&	0m7.038s	&	0m5.433s	\\	
	&	5	&	unsat	&	0m0.702s	&	0m0.477s	&	\TO	&	\TO	\\	
	&	6	&	sat	&	\TO	&	\TO	&	0m8.973s	&	0m9.323s	\\	
	&	10	&	sat	&	\TO	&	\TO	&	0m37.944s	&	0m19.580s	\\	
	&	11	&	unsat	&	1m32.664s	&	0m30.912s	&	\TO	&	\TO	\\	
	&	14	&	sat	&	\TO	&	\TO	&	2m13.661s	&	\TO	\\	
	&	40	&	sat	&	\TO	&	\TO	&	\TO	&	\TO	\\	\hsep
alternating-paths	&	2	&	sat	&	0m0.049s	&	\TO	&	0m11.144s	&	0m1.105s	\\	
	&	100	&	sat	&	\TO	&	\TO	&	\TO	&	0m9.671s	\\	\hsep
alternating-simple-paths	&	3	&	sat	&	\TO	&	\TO	&	\TO	&	0m6.754s	\\	
	&	4	&	sat	&	\TO	&	\TO	&	\TO	&	0m10.128s	\\	
	&	7	&	sat	&	\TO	&	\TO	&	\TO	&	\TO	\\	
	&	10	&	sat	&	\TO	&	\TO	&	\TO	&	\TO	\\	\hsep
exponential	&	3	&	sat	&	\TO	&	\TO	&	0m12.255s	&	0m1.847s	\\	
	&	4	&	sat	&	\TO	&	\TO	&	0m15.358s	&	11m6.482s	\\	\hsep
one-var-alternating-sat	&	300	&	sat	&	0m0.037s	&	0m0.497s	&	0m11.605s	&	0m9.720s	\\	\hsep
one-var-alternating-unsat	&	5	&	unsat	&	0m0.026s	&	0m0.073s	&	0m22.537s	&	0m54.198s	\\	\hsep
one-var-nested-exists-sat	&	300	&	sat	&	0m0.031s	&	0m0.045s	&	0m7.132s	&	0m0.562s	\\	\hsep
one-var-nested-forall-sat	&	500	&	sat	&	0m0.033s	&	\TO	&	0m7.183s	&	0m7.318s	\\	\hsep
path-unsat	&	2	&	unsat	&	0m0.033s	&	0m0.044s	&	\TO	&	1m37.099s	\\	
	&	3	&	unsat	&	0m0.030s	&	0m0.062s	&	\TO	&	1m35.451s	\\	
	&	6	&	unsat	&	0m0.037s	&	0m0.891s	&	\TO	&	1m39.209s	\\	
    \bottomrule
    \end{tabular}%
    \caption{Running time comparison on artificial benchmarks}
  \label{tab:artificial}%
\end{table}

%----------------------------------------------------
\section{Discussion}\label{se:discussion}\label{se:future}
\mypara{Related work}
Verification of database-centric software systems has received increasing attention in recent years~\cite{DBLP:journals/sigmod/DeutschHV14}. 
Tools from program analysis and model-checking 
are used to reason about the correctness of programs which access a database. Unlike our approach,
the services accessing the database are usually provided \emph{a priori} in terms of a specification in the style of a local contract~\cite{Nigam,Kumaran}.  
The code of the services themselves may be automatically synthesized from the specification, cf.\ e.g.~\cite{DBLP:journals/vldb/FernandezFLS00,Florescu00weave:a,DBLP:journals/sigmod/DeutschHV14,Deutsch:2005:VID:1066157.1066219}.
The focus of verification then is on global temporal properties of the system assuming the local contracts. 
In contrast, out goal is to verify that the input code (written by a programmer rather than generated automatically)
is correct with respect to a local specification. We discuss this also in Section~\ref{se:future}.

Several papers use variations of FO\two\ to study verification of programs that manipulate relational information.
\cite{calvanese2014shape} presents a verification methodology based on FO\two, a description logic and a separation logic for analyzing the shapes and content of 
in-memory data structures. 
\cite{Rensink} develops a logic similar to FO\two\ to reason about shapes of data structures. In both~\cite{calvanese2014shape} and~\cite{Rensink},
the focus is on analysis of shapes in dynamically-allocated memory, and databases are not studied. Furthermore, no tools based on these works are available. 
A description logic related to FO\two\ was used in~\cite{calvanese2013evolving} to verify 
that graph databases preserve the satisfaction of constraints
as they evolve. The focus of this work is on the correctness of the database, rather than the programs manipulating it. 
The verification method suggested was not implemented. In fact, to our knowledge no description logic solver implements reasoning tasks for
the description logic counterpart of FO\two\ studied in~\cite{calvanese2013evolving}, not even solvers for expressive description logics such as SROIQ. 
%Our problem has a mathematical kinship to shape analysis: In both cases, the goal is to verify programs that manipulate a relational structure. 
%Our advantage is that SQL is much more predictable than programs with pointers.

Verification of script programs with embedded queries has revolved around security, see~\cite{felderer2015security}. 
However, it seems
no other work has been done on such programs. 
%. For instance,~\cite{stringsPHP2,stringsPHP} study 
%verification of string operations in PHP. 
%~\cite{diaz2006verification} studies the verification 
% of web services in terms of properties relating to time intervals. 

\mypara{Conclusion and future work}
We developed a verification methodology for script programs with access to a relational database via SQL. 
We isolated a simple but useful fragment \sqlFragment\ of SQL and developed a simple script programming language \progLang\
on top of it. We have shown that verifying the correctness of \progLang\ programs with respect to specifications
in FO\twoBD\ is decidable. 
We implemented a solver for the 
FO\two\ finite satisfiability problem, and, based on it, a verification tool for \progLang\ programs.
Our experimental results are very promising and suggest that our approach has great potential 
to evolve into a mainstream method for the verification of script programs with embedded SQL statements. 
%%%%%%%%%%%%% BUG 
% We have demonstrated the applicability of our verification tool by applying it to Panda web administrator,
% an open source mailing list administration tool. Our verification tool has detected an error in the newsletter (un)subscription
% process. We have verified the correctness of a corrected version of Panda.  

While we believe that many of the SQL statements that appear in real-life programs fall into our fragment \sqlFragment\, 
it is evident that future tools need to consider all of database usage in real-world programs. 
In future work, we will 
explore the extension of \progLang\ and \sqlFragment. Our next goal is to be able to verify large, real-life script programs such as Moodle~\cite{moodle},
whose programming language and SQL statements use e.g.\ some arithmetic or simple inner joins. To do so, we will adapt our
approach from the custom-made syntax of \progLang\ to a fragment of PHP. 
We will both explore decidable logics extending FO\twoBD, and investigate verification techniques  based on undecidable logics
including the use of first-order theorem provers such as Vampire~\cite{riazanov2002design,kovacs2013first}
and abstraction techniques which guarantee soundness but may result in spurious errors~\cite{CEGAR}.
For dealing with queries with transitive closure, it is natural to consider fragments of Datalog~\cite{ceri1989you}. 

A natural extension is to consider global temporal specifications in addition to local contracts.
Here the goal is to verify properties of the system which can be expressed in a temporal logic such as Linear Temporal Logic LTL~\cite{pnueli1977temporal,clarke1986automatic}.
The approach surveyed in~\cite{DBLP:journals/sigmod/DeutschHV14}, which explore global temporal specifications of services given in terms of local contracts, 
may be a good basis for studying global temporal specifications in our context.

Another research direction which emerges from the experiments in Section~\ref{se:experiments} is to explore 
how to improve the performance of our FO\two\ solver on unsatisfiable inputs. 

% \item We are currently studying a large Java-based university information system with embedded SQL. This will help us to learn about possible extensions of our framework to Java.
% 
% \item To move this research area ahead, a methodological collaboration with database research is necessary.

%----------------------------------------------------

% \subparagraph*{Acknowledgements}
% 
% I want to thank \dots

\newpage

\appendix
\section{PANDA web administrator}
\subsection{PANDA Source Code: confirm.php}\label{app:panda:confirm}
\lstset{
  language        = php,
  basicstyle      = \small\ttfamily,
%   keywordstyle    = \color{dkblue},
%   stringstyle     = \color{black},
%   identifierstyle = \color{dkgreen},
   commentstyle    = \color{dark-gray},
%   emph            =[1]{php},
%   emphstyle       =[1]\color{black},
%   emph            =[2]{if,and,or,else},
%   emphstyle       =[2]\color{dkyellow},
  backgroundcolor = \color{white},
  showstringspaces=false,
  }
  
We present the code of {\tt newsletters/confirm.php} from PANDA Web Administrator version 1.0rc2 in Fig.~\ref{app:fig:php}. 
The code was translated manually to the confirm function in Fig.~\ref{fig:code}. 
The biggest difference between the PHP code and the \progLang\ code
is that the PHP code uses
{\tt dbh->getRow} to perform an SQL query which returns one row, whereas in \progLang\ this is divided into two steps:
first a SELECT query is executed and then CHOOSE selects one row. 
Additionally, the PHP code performs  some more sanity checks, 
\begin{figure}[h]
\begin{lstlisting}
<?php
// No code supplied? Why are you calling us?
if(!$_GET['code']){
	print "Code not specified."; die;
}

// Include common stuff
require_once('common.php');

// Fetch code from command line
$vcode = $_GET['code'];

// Check if code exists in DB
$row = $dbh->getRow('SELECT * FROM newsletter_addresses 
                     WHERE confirm_code = ?', array(md5($vcode)), 
                     DB_FETCHMODE_ASSOC);

// If it doesn't exist, die
if(!$row['confirm_code']){
	print "No such code"; die;
}

// If user is not subscribed, 
// then the confirmation is to subscribe him, so do it
if($row['subscribed'] == 'f'){
	if(!$dbh->query('UPDATE newsletter_addresses 
	                 SET subscribed = TRUE, confirm_code = NULL 
	                 WHERE confirm_code = ?',array(md5($vcode)))){
		print "Error while accessing database, 
		       contact system administrator."; die;
	};
} else { // Else, code is to unsubscribe him, so do it
	if(!$dbh->query('DELETE FROM newsletter_addresses 
	                 WHERE confirm_code = ?',array(md5($vcode)))){
		print "Error while accessing database, 
		       contact system administrator."; die;
	}
}
print "TRUE";
?>
\end{lstlisting}
\caption{\label{app:fig:php}The code of confirm.php, on which \confirm\ in Fig.~\ref{fig:code} is based. }
\end{figure}

\normalsize
\subsection{Correcting the Error in Panda Source Code}

In Section~\ref{se:running} we described a natural correction of the error in the running example. 
Under this correction, $\tt \confirm-corrected$ satisfies the pre- and postconditions $\pre_\confirm$ and $\post_\confirm$ from Fig.~\ref{fig:conditions}. 

\ \\
\small
{\tt
\confirm-corrected(\cd,$\tt \subscribeghost$): \\
\hspace*{0.5cm} A = \SELECT \subscribe \FROM \nwlEmail \WHERE \confirmcode\ = \cd;\\
\hspace*{0.5cm} if (A = empty) exit;  {\color{dark-gray}//"No such code"} \\
\hspace*{0.5cm} if ($\tt \subscribeghost$ = false) 
        \UPDATE \nwlEmail \SET \subscribed\ = true, \confirmcode\ = $\nil$\WHERE \confirmcode\ = \cd\\
\hspace*{0.5cm}  else\ \DELETE \FROM \nwlEmail \WHERE \confirmcode\ = \cd;\\
}

$\subscribeghost$ is no longer a ghost variable. Now it is a second argument to $\tt \confirm-corrected$. The function \subscribe,
which had generated URLs calling $\confirm$ with one argument, namely the confirm code, now generates URLs 
with an additional argument $\mathit{true}$. Similarly,  \unsubscribe\ generates URLs with the additional argument $\mathit{false}$.

\section{The Conference Management Example}\label{app:conference}

% \subsection{Firewall}
% Incorrect:

% (Devices 1) (CanSend 2 BOUNDED 2)
% requires forall d1 d2 b (CanSend(d1,d2,b) -> Devices(d1)) &
% forall d1 d2 b (CanSend(d1,d2,b) -> Devices(d2)) &
% forall d1 (Devices(d1)->(exists d2 b (CanSend(d2,d1,b)))) &
% exists d1 (Devices(d1) & (forall d2 b (Devices(d2) -> (CanSend(d2,d1,b)))))
% DELETE FROM CanSend WHERE Column1=device_to_delete OR Column2=device_to_delete
% ensures exists d2 b (Devices(d2) & (forall d1 ((~(d1 = device_to_delete) & Devices(d1)) -> (CanSend(d1,d2,b)))))
% 
% 
% Correct:
% 
% (Devices 1) (CanSend 2)
% requires exists d2 ((~(d2 = device_to_delete)) && (Devices(d2)) && forall d1 (Devices(d1) -> CanSend(d1,d2)))
% DELETE FROM CanSend WHERE Column1=device_to_delete OR Column2=device_to_delete
% ensures exists d2 (Devices(d2) &&
% forall d1 ((~(d1 = device_to_delete) && Devices(d1)) -> CanSend(d1,d2)))

In this example we verify parts of a system for conference management
which assigns reviewers to papers and records the reviews and acceptance/rejection decisions.
We focus on the earlier parts of the reviewing process:
First, potential reviewers (e.g.\ PC members) bid on papers to review. 
Based on the bids, reviewers are assigned to the papers (e.g.\ by the PC chair). 
An additional functionality of the system that we focus on is displaying the list of 
papers by a specific author. 

\subsection{The database}
The database contains the following tables and columns:
\begin{description}
 \item[Papers] with columns $\mathit{paperId}$, $\mathit{status}$, and $\mathit{session}$. The column
 $\mathit{status}$ is over the bounded domain consisting of $\mathit{undecided}$, $\mathit{accepted}$, 
 or $\mathit{rejected}$.
 The column $\mathit{session}$ ranges over the bounded domain consisting of
 $\mathit{null}$,$\mathit{blank}$,$\mathit{invited}$,$1$,$\ldots$,$k$;
 \item[PaperAuthor] with columns $\mathit{userId}$ and $\mathit{paperId}$; 
 \item[ReviewerBids] with columns $\mathit{userId}$ and $\mathit{paperId}$;
 \item[ReviewerAssignments] with columns $\mathit{userId}$ and $\mathit{paperId}$;
 \item[Conflicts] with columns $\mathit{userId}$ and $\mathit{paperId}$. 
\end{description}
The columns $\mathit{userId}$ and $\mathit{paperId}$ range over the unbounded domain. 
The key of Papers is $\mathit{paperId}$. The other tables have a many to many relationship between $\mathit{userId}$
and $\mathit{paperId}$ attesting respectively to the fact that
the user is the author of the paper, the user has bid to review the paper, the user has been assigned to review the paper,
or the user is in conflict with the paper (and therefore cannot review it). 

Before the bidding process begins, all papers are assigned the status $\mathit{undecided}$ and the 
session $\mathit{invited}$ (for an invited paper) or $\mathit{null}$ (for a contributed submission). 
The session value $\mathit{blank}$ comes up in the display function, and is at the root of a bug in the program. 

\subsection{The functions \texttt{bid}, \texttt{assign}, and \texttt{display}}

The functions \texttt{bid}, \texttt{assign}, and \texttt{display} are referred to
as \texttt{conf-bid}, \texttt{conf-assign}, and \texttt{conf-display} in Table~\ref{tab:smallExamples}. 
The code of the functions \texttt{bid}, \texttt{assign}, and \texttt{display} can be found
in Fig.~\ref{fig:conference-code}. 
The function \texttt{bid} registers that the user $\mathit{usr}$ is willing to review the paper $\mathit{ppr}$
with the sanity check that there is no conflict between the user and the paper. 
The table $A$ is either empty whenever no conflict is found, or contains the single row $ppr$
when there is a conflict. 
The function \texttt{assign} registers that the user $\mathit{usr}$ 
is assigned to review the paper $\mathit{ppr}$. 
The function \texttt{display} receives as input the user id $\mathit{usr}$
and returns the list of papers by $\mathit{usr}$ that should be displayed. 
If the review phase of the conference is not yet completed 
(i.e.\ the Boolean argument $\mathit{stillReviewing}$ has value true),
\texttt{display} removes the session values of contributed papers from the output. 
This is done to prevent leaking the information
that a contributed paper has been accepted (since only accepted papers have sessions) before the status of the paper has been announced. 
\texttt{display} leaves the status value $\mathit{invited}$ visible. 

We present two versions of \texttt{display}: one correct and one incorrect. 
\texttt{display-incorrect} leaves the session value $\mathit{null}$ unchanged. Since $\mathit{null}$
and $\mathit{blank}$
are different values, 
the information leak which the program tries to avoid is still present. 
The correct version \texttt{display-correct} differs from \texttt{display-incorrect} by
also replacing the status $\mathit{null}$ by $\mathit{blank}$. This is done by expanding the WHERE condition of the UPDATE.

\begin{figure}\tt 
\footnotesize
\noindent
bid(usr, ppr): \\
\hspace*{0.5cm} A = \SELECT paperId\FROM Papers\WHERE paperId = ppr\AND 
\NOT (paperId IN 
\\ \hspace*{0.5cm} \hspace*{0.5cm} 
(\SELECT paperId\FROM Conflicts\WHERE userId = usr));\\
\hspace*{0.5cm} if (A = empty) exit;\\
 \hspace*{0.5cm} \INSERT (usr, ppr)\INTO ReviewerBids\\
\ \\ 
\ \\
assign(usr,ppr): \\
\hspace*{0.5cm} \INSERT (usr, ppr)\INTO ReviewerAssignments\\
\ \\ 
\ \\
display-incorrect(usr, stillReviewing):\\
\hspace*{0.5cm} Output = \SELECT *\FROM Papers\WHERE paperId\IN \\
 \hspace*{0.5cm} \hspace*{0.5cm}  
 (\SELECT paperId\FROM  PaperAuthor\WHERE userId = usr);\\
\hspace*{0.5cm} if (stillReviewing = false) exit; \\
\hspace*{0.5cm} \UPDATE Output\SET session = blank\WHERE session\IN 1,...,k \\ 
\ \\ 
\ \\
display-correct(usr, stillReviewing):\\
\hspace*{0.5cm} Output = \SELECT *\FROM Papers\WHERE paperId\IN \\
 \hspace*{0.5cm} \hspace*{0.5cm}  
 (\SELECT paperId\FROM  PaperAuthor\WHERE userId = usr);\\
\hspace*{0.5cm} if (stillReviewing = false) exit; \\
\hspace*{0.5cm} \UPDATE Output\SET session = blank\WHERE session\IN 1,...,k,null \\ 
\caption{\label{fig:conference-code} The code of \texttt{bid}, \texttt{assign}, \texttt{display-incorrect}, and \texttt{display-correct}. }
\end{figure}

\subsection{The specification}

The database preserves two invariants:
\[
 \begin{array}{lll}
 \mathit{Inv}_1 & = & \forall_{\Domshort} x, y. \mathit{ReviewerBids}(x,y) \to \neg\mathit{Conflicts}(x,y)\\
 \mathit{Inv}_2 & = & \forall_{\Domshort} x, y. \mathit{ReviewerAssignments}(x,y) \to \neg\mathit{Conflicts}(x,y) 
 \end{array}
\]
These invariants state that no user may bid or be assigned to review a paper with which they are in conflict. 

The specification of \texttt{display} is as follows:
\[
\begin{array}{lll}
\pre_{\mathtt{display}} & = & \mathit{Inv}_1 \land \mathit{Inv}_2 \\
\post_{\mathtt{display}} & = & \mathit{Inv}_1 \land \mathit{Inv}_2 \land \mathit{no\mbox{-}leak}\\
\mathit{no\mbox{-}leak} &=& \mathit{stillReviewing} \to 
\left(\forall_{\Domshort} x,y.\, \mathit{Output}(x,y) \to (y = \mathit{blank} \lor y = \mathit{invited})\right)
\end{array}
\]
This specification holds for \texttt{display-correct} and does not hold for \texttt{display-incorrect}. 

For \texttt{bid} and \texttt{assign} we provide two specifications, one correct and one incorrect. 
The correct specification is as follows: 
\[
 \begin{array}{lll}
 \pre_{\mathtt{bid}}^{c}  = 
 \post_{\mathtt{bid}}^{c} = \post_{\mathtt{assign}}^{c} = \mathit{Inv_1}\land\mathit{Inv_2} \\ 
 \pre_{\mathtt{assign}}^{c}  = \mathit{Inv_1}\land\mathit{Inv_2} \land \mathit{ReviewerBids(usr,ppr)} 
 \end{array}
\]
In order to ensure that $\mathit{Inv}_2$ is preserved by \texttt{assign}, we only allow a reviewer assignment
to occur if there was a corresponding reviewer bid. Reviewer bids are required to avoid the conflicts by $\mathit{Inv}_1$,
and thus $\mathit{Inv}_2$ is preserved. 
The incorrect specification for \texttt{bid} and \texttt{assign} is as follows:
\[
 \begin{array}{lll}
 \pre_{\mathtt{bid}}^{\mathit{ic}}  &=& \mathit{Inv}_2 \\ 
 \pre_{\mathtt{assign}}^{\mathit{ic}}  &=& \mathit{Inv_2} \land \mathit{ReviewerBids(usr,ppr)}  \\ 
 \post_{\mathtt{bid}} &=& \post_{\mathtt{assign}} = \mathit{Inv_1}\land\mathit{Inv_2}
 \end{array}
\]
It is obtained from the correct specification by  omitting $\mathit{Inv}_1$ from
the preconditions. 

\section{The Firewall Example}\label{app:firewall}\label{app:network}

In this example we verify a simple firewall with respect to a simple invariant. The firewall is provided with
a database consisting of two tables: $\mathit{Device}$ and $\mathit{CanSend}$.
The table $\mathit{Device}$ consists of a single column $\mathit{deviceId}$. The table $\mathit{CanSend}$
consists of two columns $\mathit{senderId}$ and $\mathit{receiverId}$. 
The $\mathit{CanSend}$ table determines whether device A is allowed to send to device B.

The global invariant of the firewall requires that there exists a device to which every other device can send:
\[
 \mathit{Inv}_{\mathtt{firewall}} = \exists x \left(\mathit{Device}(x) \land  \forall y (\mathit{Device}(y) \to  \mathit{CanSend}(y,x))\right)
\]
We want to verify that this invariant holds when the network topology is changed. We consider the function \texttt{delete-device}: 

\ \\ 
{\tt 
\footnotesize
\noindent
delete-device(deviceToDelete): \\
\hspace*{0.5cm} \DELETE\FROM CanSend\WHERE senderId=deviceToDelete\ OR receiverId=deviceToDelete \\ 
\hspace*{0.5cm} \DELETE\FROM Device\WHERE deviceId=deviceToDelete
}
\ \\ 

Table~\ref{tab:smallExamples} refers to \texttt{delete-device}
in the rows labeled $\mathit{firewall}$. In this table, we experiment with two specifications. 
The incorrect specification is:
\[\{\mathit{Inv}_{\mathtt{firewall}}\}\ \mathtt{delete-device}\ \{\mathit{Inv}_{\mathtt{firewall}}\}\]
This specification is incorrect since it is possible that the only device which can receive messages from 
all other devices is exactly the device $\mathit{deviceToDelete}$ removed by \texttt{delete-device}. 
Our correct specification is:
\[\{\pre_{\mathtt{delete-device}} \land \mathit{Inv}_{\mathtt{firewall}}\}\ \mathtt{delete-device}\ \{\mathit{Inv}_{\mathtt{firewall}}\}\]
where
\[
 \pre_{\mathtt{delete-device}} = \exists x \left((x \not= \mathit{deviceToDelete}) 
 \land \mathit{Device}(x) \land  \forall y (\mathit{Device}(y) \to  \mathit{CanSend}(y,x))\right)
\]
This correct specification ensures that there is a device as required which is not $\mathit{deviceToDelete}$. 

\subsection{Inflated examples}\label{app:network-inflated}
As a basic test of the scalability of our approach, in Section~\ref{se:inflated} we created large examples 
by inflating small examples. Here we illustrate this on the firewall example. 
The following is the result of inflating \texttt{delete-device} with multiplier $3$. The function \texttt{delete-device3}
deletes three devices from the network:

\ \\ 
{\tt 
\footnotesize
\noindent
delete-device3(deviceToDelete1,deviceToDelete2,deviceToDelete3): \\
\hspace*{0.5cm} \DELETE\FROM CanSend\WHERE senderId=deviceToDelete1\ OR receiverId=deviceToDelete1 \\ 
\hspace*{0.5cm} \DELETE\FROM CanSend\WHERE senderId=deviceToDelete2\ OR receiverId=deviceToDelete2 \\ 
\hspace*{0.5cm} \DELETE\FROM CanSend\WHERE senderId=deviceToDelete3\ OR receiverId=deviceToDelete3 \\ 
\hspace*{0.5cm} \DELETE\FROM Device\WHERE deviceId=deviceToDelete1\\
\hspace*{0.5cm} \DELETE\FROM Device\WHERE deviceId=deviceToDelete2\\
\hspace*{0.5cm} \DELETE\FROM Device\WHERE deviceId=deviceToDelete3\\
}
\ \\ 
The specifications must be altered correspondingly, so the precondition of the correct specification is changed to:
\[
\begin{array}{lll}
 \pre_{\mathtt{delete-device3}} &=& \exists x ((x \not= \mathit{deviceToDelete}1) \land 
 (x \not= \mathit{deviceToDelete}2) \land (x \not= \mathit{deviceToDelete}3)  \\ 
 && \land \mathit{Device}(x) \land  \forall y (\mathit{Device}(y) \to  \mathit{CanSend}(y,x)))
 \end{array}
\]
The other pre- and postconditions remain unchanged, since they just consist of the invariant. 

In Table~\ref{tab:inflatedExamples}, the rows for $\mathit{firewall}$ correspond to the
inflated versions of \texttt{delete-device} with the altered specifications. 
\section{\texorpdfstring{FO\two{}}{FO2} Reasoning}\label{app:se:FO2}\label{app:se:foreasonig}

This appendix gives more detail on the implementation of our FO\two\ finite satifiability solver. 
Section~\ref{app:se:forallforall} gives explicityly the Tseitin transformation of the $\forall\forall$-conjuct from Section~\ref{se:FO2:CNF}. 
Section~\ref{app:se:efficient} discusses our use of a refined bound on the size of the maximal model
due to \cite{GKV} to improve the efficiency of the solver. 

\subsection{The \texorpdfstring{$\forall\forall$}{forall forall}-conjunct}\label{app:se:forallforall}
Here we give $B^{\forall\forall}$ explicitly. 

Let $\forall x \forall y\, \alpha'$ be the unique $\forall\forall$-conjunct of $\psi$. 
The CNF formula $B^{\forall\forall}$ will have clauses which explicitly detail the semantics of $\forall x \forall y\, \alpha'(x,y)$ on
the structure $\mathcal{A}_{\As}$. 
We associate each sub-formula $\gamma(x,y)$ in $\alpha'(x,y)$ and $\ell_1,\ell_2 \in \{1,\ldots,\bnd(\psi)\}$ with
a new Boolean variable $u^{\gamma}_{\ell_1,\ell_2}$. $B^{\forall\forall}$ will contain clauses guaranteeing that
under any assignment $\As$ satisfying $B_\psi$, 
\begin{itemize}
 \item[\ ] ($\triangle$) $S(u^{\gamma}_{\ell_1,\ell_2})=True$  iff $\mathcal{A}_{\As} \models \gamma(\ell_1,\ell_2)$.
\end{itemize}
Additionally, we add to $B^{\forall\forall}$ the clause 
$u^{\alpha'}_{\ell_1,\ell_2}$
for all $1\leq \ell_1,\ell_2 \leq \bnd(\psi)$ to assert that 
$\alpha'$ is true for all values of $x$ and $y$. 

It remains to describe the clauses which define the values of $u^{\gamma}_{\ell_1,\ell_2}$ according to ($\triangle$). We distinguish two cases, depending on whether $\gamma$
is an atom of $\psi$ or is obtained by applying a Boolean connective $\lor,\land,\neg$ on sub-formulas. Consider first the case that $\gamma(x,y) = \neg (\delta(x,y))$. 
For every $1\leq \ell_1,\ell_2\leq \bnd(\psi)$, we add the clauses 
\[\neg u^{\gamma}_{\ell_1,\ell_2} \lor \neg u^{\delta}_{\ell_1,\ell_2}\mbox{ \ \ \ and \ \ \ }
u^{\gamma}_{\ell_1,\ell_2} \lor u^{\delta}_{\ell_1,\ell_2}\]
whose conjunction is equivalent to 
$u^{\gamma}_{\ell_1,\ell_2} \leftrightarrow \neg u^{\delta}_{\ell_1,\ell_2}$. The other Boolean connectives are axiomatized similarly. 

For the case of atoms, $u^{\gamma}_{\ell_1,\ell_2}$ gets it value from one of the $v$ variables 
by choosing the indices correctly:
\begin{itemize}
 \item If $\gamma(x,y)=R(x,y)$, then for every $1\leq \ell_1,\ell_2\leq \bnd(\psi)$, $u^{R(x,y)}_{\ell_1,\ell_2}$ is assigned the same value as $v_{R,\ell_1,\ell_2}$ by adding the clauses 
 \[\begin{array}{l}
\neg u^{R(x,y)}_{\ell_1,\ell_2} \lor v_{R,\ell_1,\ell_2}\mbox{\ and \ }
u^{R(x,y)}_{\ell_1,\ell_2} \lor \neg v_{R,\ell_1,\ell_2}
\end{array}
\]
to $B_\psi$. The conjunction of these two clauses is equivalent to $u^{R(x,y)}_{\ell_1,\ell_2}  \leftrightarrow v_{R,\ell_1,\ell_2}$.
\item If $\gamma(x,y)=R(y,x)$, then $u^{R(x,y)}_{\ell_1,\ell_2}$ is assigned the same value as $v_{R,\ell_2,\ell_1}$.
\item If $\gamma(x,y)=R(x,x)$, then $u^{R(x,x)}_{\ell_1,\ell_2}$ is assigned the same value as $v_{R,\ell_1,\ell_1}$.
\item If $\gamma(x,y)=R(y,y)$, then $u^{R(y,y)}_{\ell_1,\ell_2}$ is assigned the same value as $v_{R,\ell_2,\ell_2}$.
\item If $\gamma(x,y)=U(y)$, then $u^{U(y)}_{\ell_1,\ell_2}$ is assigned the same value as $v_{U,\ell_2}$.
\item If $\gamma(x,y)=U(x)$, then $u^{U(x)}_{\ell_1,\ell_2}$ is assigned the same value as $v_{U,\ell_1}$.
\end{itemize}

\subsection{Axiomatizing models of at size at most \texorpdfstring{$\bnd(\psi)$}{bnd(psi)}}
Here we continue the discussion postponed to the appendix in  
Section~\ref{se:FO2:CNF}. Recall that
by Proposition~\ref{prop:assignment}, $B_\psi$ is satisfiable iff $\psi$ is satisfiable by a structure of cardinality  \emph{exactly} $\bnd(\psi)$.
However, \cite{GKV} guarantees only that  $\bnd(\psi)$ is an \emph{upper bound} on the cardinality of a satisfying model. 
In this appendix we explain how to construct $C_\psi$ so that it is
 satisfiable iff 
 $\psi$ is satisfiable by a structure of cardinality  \emph{at most} $\bnd(\psi)$ as follows. 

 We compute from $\psi$ a new FO\two-sentence $\Uni(\psi)$ in 
 SSNF and set $C_\psi = B_{\Uni(\psi)}$. 
 %Let $\const(\psi)$ be the set of constant names in $\psi$. 
 Let $\mathit{Uni}$ be a fresh unary relation name. 
 Let $\Uni(\psi)$ be:
 \[
 \forall x \forall y\, \left(\Uni_\alpha(x,y)
 \land \bigwedge_{i=1}^{m+1} F_i(x,y)\to\Uni_{\beta_i}(x,y)\right)
 \land \bigwedge_{i=1}^{m+1} \forall x \exists y\, F_i(x,y)
 \]
 with $\Uni_\alpha(x,y) = \mathit{Uni}(x) \land \mathit{Uni}(y) \to \alpha(x,y)$, $\Uni_{\beta_i }(x,y) = \mathit{Uni}(x) \land \mathit{Uni}(y) \to \beta_i(x,y)$, $i=1,\ldots,m$, and 
  $\Uni_{\beta_{m+1}}(x,y) =  \mathit{Uni}(y)$. 
  $F_{m+1}$ is used to guarantee that $\mathit{Uni}$ is non-empty. 
 Let $\voc(\Uni(\psi))$ be the vocabulary of $\Uni(\psi)$. 
 \begin{proposition} 
 Let $\mathcal{A}$ be a $\voc(\Uni(\psi))$-structure. Let $\mathcal{A}_{\Uni}$ be the substructure of $\mathcal{A}$ whose universe is $\mathit{Uni}^{\mathcal{A}}$.
  We have $\mathcal{A} \models \Uni(\psi)$ iff $\mathcal{A}_{\Uni} \models \psi$. 
 \end{proposition}

\subsection{An efficient finite satisfiability algorithm}\label{app:se:efficient}

The algorithm from Section~\ref{se:bounded} on which our algorithm in Section~\ref{se:sat} is based was written
from a theoretical point of view aiming to simplify the proof of the NEXPTIME-completeness of the finite satisfiability problem. 
In this section, we introduce several optimizations which, while not affecting the complexity of the problem, 
improve the performance of our satisfiability solver. 

\subsubsection{A refined upper bound}
In the course of the proof of the bounded model property,~\cite{GKV} give a more refined version of the upper bound on the size 
of a minimal satisfying model. This more refine version leads to smaller upper bounds in many cases. 
To state the refined upper bound we need some definitions. 
\begin{definition} {\bf ($1$-types and kings)}~
\begin{enumerate}
 \item A $1$-type $t(x)$ is a maximally consistent set of atomic formulas and their negations which do not have $y$ as a free variable. 
 \item For a structure $\mathcal{A}$ and an element $a$ of the universe of $\mathcal{A}$, the $1$-type of $a$ in $\mathcal{A}$
 is the unique $1$-type $t(x)$ such that $\mathcal{A}\models t(a)$. We say that $a$ \emph{realizes $t$ in $\mathcal{A}$}.  
 \item Given a structure $\mathcal{A}$ and an element $a$ of $\mathcal{A}$, $a$ is a \emph{king} in $\mathcal{A}$ if there is no other
 element in $\mathcal{A}$ with the same $1$-type as $a$. 
\end{enumerate}
\end{definition}
For example, for a vocabulary consisting of one binary relation name $R$ and one constant name $c$, the following is a $1$-type:
\[
\begin{array}{ll}
 \{ R(c,c), R(x,c), \neg R(c,x), R(x,x), \\ \neg (x = c), \neg (c = x), (x = x), (c = c) \}
 \end{array}
\]
% Every $1$-type must contain $x=x$ and $c=c$ to be consistent. Moreover, by the symmetry of $=$, a $1$-type contains $x=c$ iff it contains $c=x$. 
% Atoms which do not contain $y$ as a free variable may be of the form e.g. $R(x,x)$, $R(x,c)$, $x = c$, $U(c)$, etc. where $c$ is a constant name, $R$ is a binary
% relation name and $U$ is an unary relation name. 

\begin{lemma}[\cite{GKV}, Theorem 4.3]\label{app:lem:refined-bound}
 Let $\psi$ be a sentence in Skolemized Scott Normal Form and let $m$ be the number of conjuncts of the form $\forall x \exists y\, \beta_i$ as in Eq.~(\ref{eq:SSNF}). 
 Let $\mathcal{A}$ be a structure satisfying $\psi$. Let $K$
 be the set of all kings in $\mathcal{A}$ and let $P$ be the set of $1$-types realized in $\mathcal{A}$. There is 
 a structure of cardinality at most 
 \begin{gather}
 \label{app:eq:refined-bound}
  (m+1)|K|+3m(|P| - |K|)
 \end{gather}
 which satisfies $\mathcal{A}$. 
\end{lemma}
The bound in Lemma~\ref{app:lem:refined-bound} requires already having a model of $\psi$. However, we can use it to get a bound based on syntactic considerations only. 
For $m\geq 1$, Eq.~(\ref{app:eq:refined-bound}) can be bounded from above by the sum over all $1$-types $t(x)$, such that $t(x)$ contributes $(m+1)$ if $t(x)$ 
contains $x = c$ for some $c\in \const(\psi)$, and $3m$ otherwise. Note that the number of $1$-types containing $x=c$ for some $c\in \const(\psi)$ in any one structure
is at most $|\const(\psi)|$. Hence, Eq.~(\ref{app:eq:refined-bound}) is at most:
\begin{gather}\label{app:eq:my-bound}
\begin{array}{ll}
  \displaystyle{|\const(\psi)|\, (m+1) + \sum_{ t(x)\not \models \bigvee_{c\in\const(\psi)} x=c} 3m}
  \end{array}  
\end{gather}
Notice that this bound does not depend on whether the $1$-types in the sum
are realized in any structure. It is correct since (1) any $1$-type which implies that $x=c$ for any
constant is necessarily a king and thus contributes $m+1$, while (2) any other $1$-type, which may or may not be a king, contributes at most $\max(m+1,3m) = 3m$. 

We can now augment the algorithm in Section~\ref{se:sat} to use the refined bound from Eq.~(\ref{app:eq:my-bound}). 
Our algorithm uses Proposition~\ref{prop:skolemized-scott}
to transform a FO\two-sentence $\phi$ into Scott Normal Form, which adds a new relation name for every quantifier and every connective 
in $\phi$. This comes at a heavy cost to performance, since the number of $1$-types summed over in Eq.~(\ref{app:eq:my-bound})
 is exponential in the number of relation and constant names in Scott Normal Form of $\phi$.  
In this section we provide a more economic procedure for this purpose, which introduces new relation names only as a last resort. 

Given a FO\two-sentence $\phi$, we construct a sentence $\psi$ in Scott Normal Form
such that $\phi$ and $\psi$ are satisfiable by models of the same cardinalities. 
We construct a sequence of pairs $(\phi_k,\psi_k)$, $1\leq k \leq r$ as follows.
The sequence is built according to the process described below. The length $r$ of the sequence is
determined by applying the process until no further steps can be applied. The sequence satisfies:
\begin{itemize}
 \item[--] $\phi_0$ is the Negation Normal Form\footnote{An FO formula is in Negation Normal Form (NNF)
 if the scope of every negation symbol $\neg$ is an atom. It is well-known that for every 
 FO formula $\gamma$,  an equivalent formula $\gamma'$ in NNF can be computed in linear time. 
 } of $\phi$ and \\
$\psi_0 = \forall x \forall y\, True$,
\item[--] $\phi_r = True$ and $\psi_r = \psi$,
\item[--] $\psi_k$ is in Scott Normal Form for all $1\leq k\leq r$, and
\item[--] the sets of cardinalities of the models of $\phi_k \land \psi_k$
are equal for all $1\leq k \leq r$.
%\item[--] $\phi_{k+1}$ has less quantifiers than $\phi_k$ for all $1\leq k\leq r-1$.

\end{itemize}
Given $(\phi_k,\psi_k)$ we compute $(\phi_{k+1}, \psi_{k+1})$ iteratively  as follows:  
\begin{enumerate}
 \item\label{app:item:forall-1} If $\phi_k$  is in one of the forms:
 \[
  \begin{array}{ll}
  \forall x Q y\, \epsilon\ \ \ \ \ \ \forall x\, \epsilon\\
  \forall y Q x\, \epsilon\ \ \ \ \ \ \forall y\, \epsilon 
  \end{array}
 \]
 where $\epsilon$ is quantifier-free and $Q$ is a quantifier, i.e.\ $Q\in\{\exists,\forall\}$, then $\psi_{k+1}$ is obtained from $\psi_k$ as follows.
 If $Q = \forall$, $\psi_{k+1}$ is obtained by adding 
 $\epsilon$ as a new conjunct inside the quantifiers $\forall x \forall y $. If $Q = \exists$, $\psi_{k+1}$ is obtained by adding a new conjunct
 $\forall x \exists y \, \epsilon$ or $\forall y \exists x \, \epsilon$ to $\psi_k$. 
 We set  $\phi_{k+1}$ to True. We end the iteration by setting $r$ to $k+1$,
 \item\label{app:item:forall-2} If $\phi_k$ is a conjunction in which one of the conjuncts $\gamma$ is of one of the forms in the previous item, 
 then $\psi_{k+1}$ is obtained from $\psi_k$ as in the previous item and $\phi_{k+1}$ is obtained from $\phi_k$ by removing $\gamma$ from the conjunction. 
 \item\label{app:item:existential} If there is an existential quantifier not in the scope of any other quantifier in $\phi_k$, then $\phi_{k+1}$ is obtained by removing the existential quantifier
 and replacing all the occurrences of the quantified variable bound to this quantifier with one fresh constant name. We set $\psi_{k+1} = \psi_k$. 
 Note for the correctness of this step that we are using here that the formulas are in Negation Normal Form, i.e.\ this existential quantifier may be in the scope of the $\land$ and $\lor$ operators only. 
 \item\label{app:item:recursive}  If $\phi_k$ is of one of the following forms, or 
 if $\phi_k$ is a conjunction in which one of the conjuncts $\gamma$ is of one of the following forms:
 \[
  \begin{array}{ll}
  \forall z_1\, ((Q z_2\, \delta_1) \lor \delta_2)\\
  \forall z_1\, ((Q z_2\, \delta_1) \land \delta_2)\\
  \forall z_1\, (\delta_1 \lor (Q z_2\, \delta_2))\\
  \forall z_1\, (\delta_1 \land (Q z_2\, \delta_2))
  \end{array}
 \]
 where $Q\in \{\exists,\forall\}$, $z_1,z_2 \in \{x,y\}$ and $z_1 \not=z_2$, $\phi_{k+1}$ is obtained by taking the quantifier $Q z_2$
 out of the scope of the Boolean connective. E.g., we substitute $\forall z_1 ((Q z_2 \delta_1) \lor \delta_2)$ with 
 $\forall z_1 Q z_2 (\delta_1 \lor \delta_2)$. We set $\psi_{k+1} = \psi_k$. Note for the correctness of this step
 that the $\delta_j$ not in the scope of $Q$ does not have $z_2$ as a free variable, since we are dealing with sentences. 
 \item\label{app:item:qf} If $\phi_k$ is quantifier-free, $\psi_{k+1}$ is obtained from $\psi_k$ by adding $\phi_k$ as a new conjunct inside the quantifiers
 $\forall x \forall y $ in $\psi_k$. We set $\phi_{k+1}$ to True and end the process by setting $r=k+1$.
 \item\label{app:item:inner-most} If none of the previous items applied to $\phi_k$ in this iteration, we eliminate one 
 quantifier from $\phi_k$ in spirit of the discussion before Proposition 3.1 in~\cite{GKV}. 
 Let $\gamma$ be 
 a sub-formula of $\phi_k$ of the form $Q z\, \delta$, where $Q\in \{\exists,\forall\}$, $z\in \{x,y\}$ and $\delta$ is quantifier-free. 
 Let $\bar{z} \in \{x,y\}$ such that $z\not=\bar{z}$. 
 The sentence $\phi_{k+1}$ is obtained by substituting $\gamma$ by $E(\bar{z})$ in $\phi_k$, where $E$ is a fresh unary relation name. 
 Let $\theta = \forall \bar{z}\,((Q z \delta)\leftrightarrow E(\bar{z}))$. The sentence $\theta$ says that $E$ is interpreted as the set of 
 universe elements $u$ for which $Q z \delta$ holds.  $\theta$ is equivalent to the conjunction of $\theta_{\forall \forall}$ and $\theta_{\forall \exists}$ 
 such that $\theta_{\forall \forall}$ is of the form $\forall x \forall y\, \theta'_{\forall\forall}$  and 
 $\theta_{\forall \exists}$ is of the form $\forall x \exists y\, \theta'_{\forall\exists}$, and $\theta'_{\forall\forall}$ and $\theta'_{\forall \exists}$ 
 are quantifier-free. 
 Let $\psi_{k+1}$ be obtained by adding the conjunct $\theta_{\forall \exists}$ to $\phi_k$
 and adding $\theta'_{\forall \forall}$ as a new conjunct inside the $\forall\forall$-conjunct of $\phi_k$. 
 
\end{enumerate}

Only item~\ref{app:item:inner-most} increases the number of names, so it is only used when no other item applies. 
The procedure terminates because every item, except for item~\ref{app:item:recursive}, removes a quantifier when going from $\phi_k$ to $\phi_{k+1}$,
 and whenever item~\ref{app:item:recursive} is applied, in the next iteration either item~\ref{app:item:forall-1} or item~\ref{app:item:forall-2} will be applied. 

%  The correctness of items~\ref{item:forall-1},~\ref{item:forall-2} and~\ref{item:qf} is rather straight-forward. 
%  For item~\ref{item:recursive}, notice that since $\phi_k$ is a sentence, $z_2$ is not a free variable in $\delta_2$. This is the crucial fact to get e.g. that
%  $(Q z_2\, \delta_1) \lor \delta_2$ is logically equivalent to $Q z_2\, (\delta_1 \lor \delta_2)$. 
% 
% 
%  
%  item~\ref{item:existential}. 
% 

%  
%  Prosition ~\ref{prop:scott} follows from repeated application of the following proposition based on
% Proposition 3.1 in~\cite{GKV} and the discussion preceeding it. 
% \begin{proposition}\label{app:prop:scott-one-step}
% Let $\tau$ be a vocabulary and $\phi_1$ be a FO\two$(\tau)$-sentence. Let $\psi_1$ be a FO\two$(\tau)$-sentence in Scott Normal Form. 
% There are linear-time computable vocabulary $\sigma\supseteq \tau$ and FO\two$(\sigma)$-sentences
% $\phi_2$ and $\psi_2$ such that 
% \begin{enumerate}
%  \item $\psi_2$ is in Scott Normal Form;
%  \item $\phi_1 \land \psi_1$ and $\phi_2 \land \psi_2$ are satisfied by structures of the same cardinalities;
%  \item $\sigma$ expands $\tau$ with one fresh unary relation;
%  \item $\phi_2$ has one less quantifier than $\phi_1$;
%  \item $|\phi_2\land \psi_2| = O(1) + |\phi_1\land \psi_1|$.
% \end{enumerate}
% \end{proposition}
% 

\subsubsection{Ruling out unfeasible \texorpdfstring{$1$}{1}-types}

Up until now, we have bounded the number of $1$-types which are realized in some structure $\mathcal{A}$ satisfying $\phi$ with
the number of all $1$-types. However, it is possible to determine that some $1$-types are not feasible in any structure satisfying $\phi$
and subtract them from the upper bound. 

Recall the table $\mathit{\nwlEmail}$ from the running example in Section~\ref{se:running}. $\mathit{\nwlEmail}(x,y)$ expresses that user $y$ is subscribed to
newsletter $x$. It is natural that the requirement that users and newsletters are disjoint
\[
 \forall x\, \left((\exists y\, \mathit{\nwlEmail}(x,y)) \to (\forall y\, \neg  \mathit{\nwlEmail}(y,x))\right)
\]
 is part of the database invariant. 
Hence, any $1$-type containing $\mathit{\nwlEmail}(x,x)$  is unfeasible. 

Let $\psi$ be a sentence in Scott Normal Form such that $\psi = \psi_{\universal} \land \psi_{\existential}$,
where $\psi_{\universal} = (\forall x\forall y\,\alpha)$ and $\psi_{\existential}$ is a conjunction of terms of the form $\forall x \exists y\, \beta_i$.
Any $1$-type which is not feasible for $\psi_{\universal}$ is certainly not feasible for $\psi$. Since $\psi_{\universal}$
is a universal FO sentence, it adheres to a classical property of universal FO: the class of models of $\psi_{\universal}$ is closed under taking substructures. 
This implies that any feasible $1$-type of $\psi_{\universal}$ occurs in a structure of cardinality $1$. 
We count the number of unfeasible $1$-types which do not contain $x = c$ for any constant name by iterating over all 
such $1$-types $t(x)$ and checking satisfiability of $\forall x\, t(x) \land \psi_{\universal}$ by a structure of cardinality $1$.

\subsubsection{Improved satisfiability testing}

The size of the input and the run time of the SAT solver depend on the size of the upper bound $\bnd(\psi)$ on the cardinality of the minimal satisfying model, if such a model exists. 
To guarantee that $\psi$ is unsatisfiable, one must verify that no satisfying model up to that cardinality exists. In contrast, to guarantee that 
$\psi$ is satisfiable, it is enough to find a satisfying model, which may be much smaller than the bound. Hence, instead of performing 
the satisfiability test for models of cardinality up to $\bnd(\psi)$, we iteratively search for models whose cardinalities $1=a_0,\ldots,a_e\leq \bnd(\psi)$ increase exponentially 
up to $\bnd(\psi)$. For each $a_i$, we construct a CNF formula $C_{a_i}$ similar to $C_\psi$, only replacing the maximal cardinality $\bnd(\psi)$ with $a_i$, and apply the SAT solver
to $C_{a_i}$. We have that $a_{i} = 2 a_{i-1}$ for every $2 \leq i < e$ and $a_e = \min (\bnd(\psi), 2a_{e-1})$. 
The search continues until either we reach the index $e$ for which $a_e = \bnd(\psi)$ or $C_{a_e}$ is satisfiable, and the algorithm returns the truth-value
of $C_{a_e}$. This procedure is very often much faster for satisfiable inputs $\psi$. For unsatisfiable inputs, the run time increase is negligible, since 
the formulas $C_{a_0},\ldots,C_{a_{e-1}}$ are exponentially smaller than $C_{a_e}=C_\psi$.

% \subsection{Discussion}
% Input format to tool is standard. 
% 
% Choosing a sat solver, input format to sat solver is standardized. 
% 
% Satisfiability is easy, unsatisfiability is hard. Comparison to Z3. 
% 
% Scalability and Bottleneck: writing large text files. 

\end{document}